\documentclass[final,5p,times,twocolumn]{autart}    % Enable this line and disable the 
\usepackage{graphicx}          % Include this line if your 
\usepackage{amsmath}

\usepackage[amsmath,thmmarks,amsthm,hyperref]{ntheorem}

\usepackage{amsfonts}
\usepackage{color}                       
\usepackage[font=footnotesize]{caption}
\usepackage{wrapfig}
\usepackage{subcaption}
\usepackage{soul}
\usepackage{multirow}
\usepackage{changepage}
\usepackage{cite}
\usepackage{tabularx}
\usepackage{adjustbox}
\usepackage{mathtools}
\usepackage{amssymb}
\usepackage{verbatim}
\usepackage{comment}       
\usepackage{etoolbox}
\usepackage{booktabs}
\usepackage{enumitem}

\def\real{\mathbb{R}}
\def\integer{\mathbb{Z}}
\def\natural{\mathbb{N}}

\DeclareMathOperator{\E}{\mathbb{E}}
\newcommand{\subscr}[2]{#1_{\textup{#2}}}

\newcommand{\map}[3]{#1: #2 \rightarrow #3}
\newcommand{\union}{\operatorname{\cup}}

\newcommand\oprocendsymbol{\hbox{$\square$}}
\newcommand\oprocend{\relax\ifmmode\else\unskip\hfill\fi\oprocendsymbol}

\newcommand\bit[1]{\textit{\textbf{#1}}}

\newtheorem{theorem}{Theorem}
\newtheorem{lemma}{Lemma}

\newtheorem{remark}{Remark}

\renewcommand{\theenumi}{(\roman{enumi}}
\def \bs {\boldsymbol}
\def \mc {\mathcal}

\def\cog{\mathrm{cog}}

\def\var{\mathrm{Var}}

\allowdisplaybreaks

\begin{document}

\begin{frontmatter}
\runtitle{Optimal fidelity selection for human-in-the-loop}  % Running title for regular 

\title{Structural Properties of Optimal Fidelity Selection Policies for Human-in-the-loop Queues \thanksref{footnoteinfo}} % Title, preferably not more 

\thanks[footnoteinfo]{{A preliminary version of this work~\cite{gupta2019optimal} was presented at the 2019 American Control Conference, held in Philadelphia. We expand on the work in~\cite{gupta2019optimal} by establishing structural properties of the optimal value function and consequently the optimal fidelity selection policy. We also introduce additional numerical illustrations.} Corresponding author P. Gupta. Tel. +1-517-432-0019}

\author[MSU]{Piyush Gupta}\ead{guptapi1@msu.edu},  % Add the 
\author[MSU]{Vaibhav Srivastava}\ead{vaibhav@egr.msu.edu}          % e-mail address 

\address[MSU]{Department of Electrical and Computer Engineering, Michigan State University, East Lansing, Michigan, 48824, USA}  % Please supply                                             
          
                             % keyword list or with the 
                                          % help of the Automatica 
                                          % keyword wizard

\begin{abstract}                          % Abstract of not more than 200 words.
We study optimal fidelity selection for a human operator servicing a queue of homogeneous tasks. The agent can service a task with a normal or high fidelity level, where fidelity refers to the degree of exactness and precision while servicing the task. Therefore, high-fidelity servicing results in higher-quality service but leads to larger service times and increased operator tiredness. We treat the human cognitive state as a lumped parameter that captures psychological factors such as workload and fatigue. The operator's service time distribution depends on her cognitive dynamics and the fidelity level selected for servicing the task. Her cognitive dynamics evolve as a Markov chain in which the cognitive state increases with high probability whenever she is busy and decreases while resting. The tasks arrive according to a Poisson process and the operator is penalized at a fixed rate for each task waiting in the queue. We address the trade-off between high-quality service of the task and consequent penalty due to a subsequent increase in queue length
using a discrete-time Semi-Markov Decision Process framework. We numerically determine an optimal policy and the corresponding optimal value function. Finally, we establish the structural properties of an optimal fidelity policy and provide conditions under which the optimal policy is a threshold-based policy.
\end{abstract}

\end{frontmatter}

\section{Introduction}
Human-in-the-loop systems are pervasive in areas such as search and rescue, semi-autonomous driving, and robot-assisted surgery. Many safety-critical systems rely on human expertise to ensure safe and efficient operation. Human-robot collaboration allows for integrating human knowledge and perception skills with autonomy.
In such systems, it is often of interest to increase the ratio of robots to humans, which leads to a reduction in cost but an increase in human workload. This is detrimental to the system performance as human performance is a function of cognitive factors such as fatigue and workload. Therefore, in environments with constrained human resources, it is critical to facilitate the effective use of limited cognitive resources~\cite{KS-EF:10b}. In this work, we control the cognitive state of the human operator by optimizing the fidelity level for servicing the tasks, where fidelity refers to the degree of exactness and precision while servicing the task.\\
We study optimal fidelity selection for a human operator servicing a queue of homogeneous tasks. An example scenario is an airport security system where a human scans the luggage items with different fidelity levels. The term ``fidelity" can have different meanings based on the application. For example, in shared-control tasks such as collaborative human-robot search~\cite{nourbakhsh2005human},
fidelity could refer to the human contribution to the task as compared to autonomy. Similarly, in a dual-task paradigm such as supervising and teleoperating a team of robots~\cite{yuan2017evaluating}, servicing single versus both tasks can correspond to different fidelity levels. We incorporate human cognitive dynamics into the fidelity selection problem and study its influence on optimal policy. In particular, we show that servicing the tasks with high fidelity is not always optimal due to larger service times and increased tiredness of the human operator. In fact, we show that the optimal policy depends on the number of tasks awaiting service (queue length) as well as the cognitive state of the human operator. Our results provide insight into the efficient design of human decision support systems. \\
For servicing each task, the human operator receives a reward based on the fidelity level selected for the task. However, with higher fidelity, the cognitive state quickly rises to higher sub-optimal levels, thereby requiring larger service time for subsequent tasks. Hence, there is a trade-off between the reward obtained by high-fidelity servicing (improved service quality), and the penalty incurred due to the resulting delay in servicing subsequent tasks. We elucidate this trade-off and find an optimal fidelity selection policy. Indeed the optimal policy is problem-specific and depends on the problem parameters. Therefore, without careful system design and parameter tuning such as selecting arrival rates, the optimal policy might behave unexpectedly. This can lead to a bad user experience for the human operator or a lack of trust in the optimal recommendations, for example, in a scenario where the decision-support system recommends frequent switching of the fidelity level. To this end, we establish structural properties~\cite{agarwal2008structural} of the optimal fidelity selection policy and provide conditions under which, for each cognitive state, there exist thresholds on queue lengths at which optimal policy switches fidelity levels. These structural properties can be used to tune the decision support system parameters such that the optimal policy is well-behaved and the human operator can trust its suggestions. Furthermore, these properties can be leveraged to determine a minimally parameterized policy for specific individuals which can be refined in real-time using a small amount of data.\\
In our setup, the human operator has a unimodal
performance (characterized by its service time) w.r.t. its cognitive state which is inspired by the Yerkes-Dodson law~\cite{yerkes1908relation}. Intuitively, such unimodal behavior is obtained because excessive stress (high cognitive state) overwhelms the operator and too little stress (low cognitive state) leads to boredom and a reduction in vigilance.\\
While human-in-the-loop is used as a primary
application, this work is applicable to other non-human servers with state-dependent unimodal performance. For example, in the context of traffic flow, the traffic intersection can be interpreted as a server, and traffic flux is a unimodal function of the traffic density~\cite{keyvan2012exploiting}. In such a scenario, the control measures may include admitting a vehicle or rerouting it, to maintain the optimal performance of the traffic network. \\
The major contributions of this work are threefold: (i) we pose the fidelity selection problem in  a Semi-Markov Decision Process (SMDP) framework and compute an optimal policy,
(ii) we numerically show the influence of cognitive dynamics on the optimal policy, and
(iii) we establish structural properties of the optimal fidelity policy and provide sufficient conditions for a threshold-based policy to be optimal.\\
The rest of the paper is structured in the following way. In Section~\ref{Related_work}, we discuss some relevant literature. Section~\ref{Problem Setup} presents the problem setup and formulates the fidelity selection problem using an SMDP framework. In Section~\ref{Numerical Illustration}, we numerically illustrate an optimal fidelity selection policy and establish its structural properties in Section~\ref{Structural Properties}. Finally, in Section~\ref{Conclusions}, we provide conclusions and discuss the future directions of this work.

\section{Related Work}\label{Related_work}
Recent years have seen significant efforts in integrating human knowledge and perception skills with autonomy~\cite{JP-VS-etal:12t}. A key research theme within this area concerns the systematic allocation of human cognitive resources for efficient performance. Therein, some of the fundamental questions studied include optimal scheduling of the tasks to be serviced by the operator~\cite{peters2018robust}, enabling shorter operator reaction times by controlling the task release~\cite{KS-EF:10b},  and determining optimal operator attention allocation~\cite{VS-RC-CL-FB:11zc}. In contrast to the aforementioned works, we consider an SMDP formulation to deal with general (non-memoryless) service time distributions of the human operator. Furthermore, while the above works propose heuristic algorithms, we focus on establishing the structural properties of the optimal policy.\\    
Some interesting recent studies with state-dependent queues are considered in~\cite{lin2021stabilizing, lin2020queueing}. In these works, authors design scheduling policies that stabilize a queueing system and decrease the utilization rate of a non-preemptive server that measures the proportion of time the server is working. The performance of the server degrades with the increase in server utilization and improves when the server is allowed to rest. In contrast to monotonic server performance with the utilization rate in~\cite{lin2021stabilizing, lin2020queueing}, we model the service time of the human operator as a unimodal function of its cognitive state. Our model for service time is inspired by experimental psychology literature~\cite{yerkes1908relation} and incorporates the influence of cognitive state and fidelity level on service time.\\ 
The optimal control of queueing systems~\cite{gupta2021robust} is a classical problem in queueing theory.  Of particular interest are the works~\cite{stidham1989monotonic,sennott1989average}, where authors study the optimal policies for an M/G/1 queue by SMDP formulation and describe its qualitative features. In contrast to a standard control of queues problem, the server in our problem is a human operator with cognitive dynamics that must be incorporated into the problem formulation. \\
Our mathematical techniques to establish the structural properties of the optimal policy are similar to~\cite{agarwal2008structural}. In~\cite{agarwal2008structural}, the authors establish structural properties of an optimal transmission policy for transmitting packets over a point-to-point channel in communication networks. The optimal policy of their Markov decision process depends on the queue length, the number of packet arrivals, and the channel fading state. In~\cite{yang2013structural}, authors 
study structural properties of the optimal resource allocation policy for a single-queue system in which a central decision-maker assigns servers to each job.  In contrast to~\cite{agarwal2008structural, yang2013structural}, a major challenge in our problem arises due to SMDP formulation for non-memoryless service time distribution and its unimodal dependence on the cognitive state. 

\section{Background and Problem Formulation} \label{Problem Setup}
We now discuss our problem setup, formulate it as an SMDP, and solve it to obtain an optimal policy.
\subsection{Problem Setup}
We consider a human supervisory control system in which a human operator is servicing a stream of homogeneous tasks. The human operator may service these tasks with different levels of fidelity. The servicing time of the operator depends on the fidelity level with which she services the task as well as her cognitive state. We assume that the mean service time of the operator increases with the selected fidelity level. For example, when the operator services the task with high fidelity, she may look into deeper details of the task, and consequently take a longer time to service.\\
\begin{figure}
	\centering
	\includegraphics[width=0.7\linewidth, height=\linewidth, keepaspectratio]{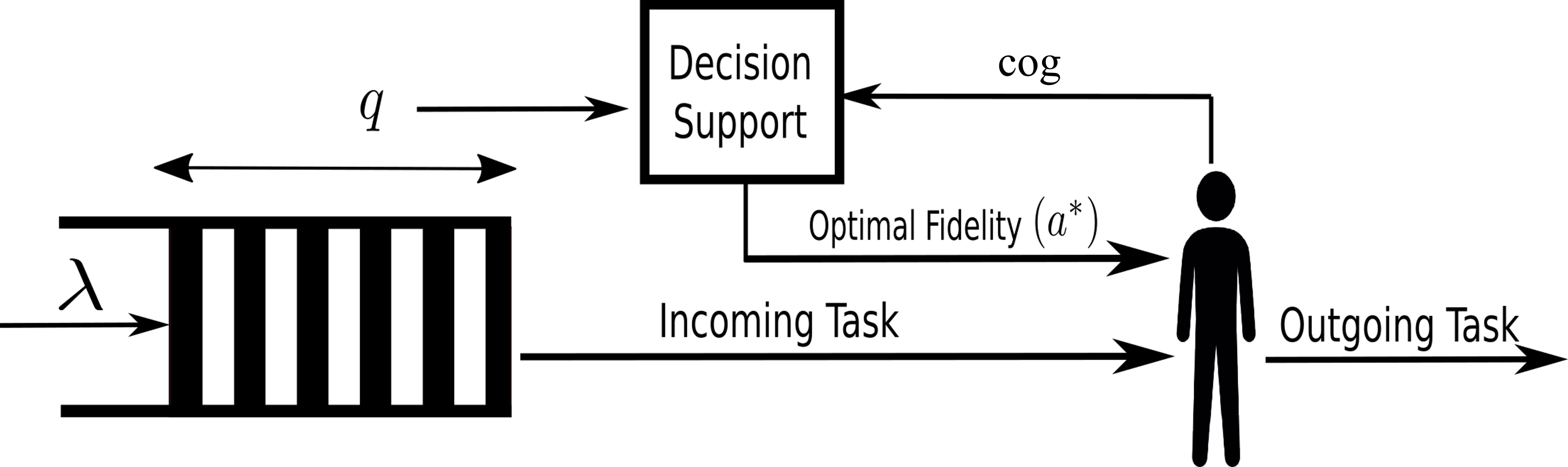}
    \caption{\footnotesize{Overall schematic of the problem setup. The incoming tasks arrive as a Poisson process with a rate $\lambda$. The tasks are serviced by the human operator based on the recommended fidelity level by the decision support system. 
    % Each task loses its value at a fixed rate while waiting in the queue.
    }} 
    \label{fig:Model}
\end{figure} 
In addition to the fidelity level, the human service time may depend on their cognitive state. We treat the cognitive state as a lumped parameter that can capture various physiological measures. It can be a function of stress, workload, arousal rate, operator utilization ratio, etc. Such lumped representation can be obtained by classifying these psychological measurements into different service time distribution parameters. Inspired by the Yerkes-Dodson law, for a fixed level of fidelity, we model the service time as a unimodal function of the human cognitive state. Specifically, the mean service time is minimal corresponding to an intermediate optimal cognitive state (later referred to as the optimal cognitive state $\cog^*$) as shown in Fig.~\ref{fig:mean_sojurn_cog}. \\
We are interested in the optimal fidelity selection policy for the human operator. To this end, we formulate a control of queue problem, where in contrast to a standard queue, the server is a human operator with her cognitive dynamics. The incoming tasks arrive according to a Poisson process at a given rate $\lambda \in \mathbb{R}_{>0}$ and are serviced by the operator based on the fidelity level recommended by a decision support system (Fig.~\ref{fig:Model}). We consider a dynamic queue of homogeneous tasks with a maximum capacity $L \in \natural$. The operator is penalized for each task waiting in the queue at a constant rate $c \in \mathbb{R}_{>0}$ per unit delay in its servicing. The set of possible actions available for the operator corresponds to (i) \textbf{\textit{Waiting (W)}} when the queue is empty, (ii) \textbf{\textit{Resting (R)}}, which allows the operator to rest and reach the optimal cognitive state, (iii) \textbf{\textit{Skipping (S)}}, which allows the operator to skip a task to reduce the queue length and thereby focus on newer tasks, (iv) \textbf{\textit{Normal Fidelity (N)}} for servicing the task with normal fidelity, and (v) \textbf{\textit{High Fidelity (H)}} for servicing the task more carefully with high precision. The skipping action ensures the stability of the queue by allowing the operator to reduce the queue length by skipping some tasks. Ideally, through appropriate control of the arrival rate, the system designer should ensure that skipping is not an optimal action.\\
 Let $s \in \mc{S}$ be the state of the system and $\mc A_s$ be the set of admissible actions in state $s$, which we define formally in Section~\ref{Mathematical}. The human receives a reward $\map{r}{\mc{S}\times \mc A_s}{\mathbb{R}_{\geq 0}}$ defined by 
\begin{equation}\label{eq:im-rew}
r(s, a) = \begin{cases}
r_H, & \text{if } a = \textit{H},\\
r_N, & \text{if } a = \textit{N},\\
0, & \text{if } a \in \{\textit{W, R, S}\},\\
\end{cases}    
\end{equation}
where, $r_H, r_N \in \real_{\ge 0}$ and $r_H > r_N$.  
We intend to design a decision support system that assists the operator by recommending optimal fidelity level to service each task\footnote{We assume compliance of the operator with the recommendations.  To account for non-compliance, we can introduce $p$ as the probability of compliance and $1-p$ as the probability that the operator will deviate and follow a different behavioral policy. This deviation can be incorporated by using a mixed service time distribution with probabilities $p$ and $1-p$ for the recommended and behavioral actions respectively.}. The recommendation is based on the queue length and the operator's cognitive state which we assume to have real-time access using, e.g., Electroencephalogram (EEG) measurements (see~\cite{rao2013brain} for measures of cognitive load from EEG data) or eye-tracking and pupillometry~\cite{palinko2010estimating}. We assume that the noisy data from these devices can be clustered into a finite number of bins to estimate the cognitive state. We study the optimal policy under the perfect knowledge of the cognitive state\footnote{If the cognitive state is not perfectly known, then our policy can be used within algorithms such as $\subscr{Q}{MDP}$~\cite{littman1995learning}, to derive approximate solutions to the associated partially observable Markov decision process~\cite{spaan2012partially}.}.
\subsection{Mathematical Modeling} \label{Mathematical}
We formulate the control of queue problem as a discrete-time SMDP $\Gamma$ defined by the following six components:
\begin{enumerate}%[leftmargin=*]
\item[(i)] A finite state space $\mc{S}:=\{(q,\cog) | \ q\in \{0,1,...,L\},$ $ \ \cog \in \mc{C}:= \{i/N\}_{i \in \{0,\cdots, N\}}\}$, for some $N \in \natural$, where $q$ is the queue length and $\cog$ represents the lumped cognitive state, which increases (decreases) when the operator is busy (idle).
\item[(ii)] A set of admissible actions $\mc A_{s}$ for each state $s \in \mc{S}$ which is given by: (i) $\mc A_s := \{\textit{W} \;| \  s \in \mc{S}, \  q=0 \}$ when queue is empty, (ii) $\mc A_s := \{\{$\textit{R}, \textit{S}, \textit{N}, \textit{H} \}$| \ s \in \mc{S}, \ q\neq 0 \}$ when queue is non-empty and $\cog > \cog^*$, where $\cog^* \in \mc{C}$ is the optimal cognitive state associated with minimum mean service time, and (iii) $\mc A_s := \{\{$\textit{S}, \textit{N}, \textit{H} \}$| \ s \in \mc{S}, \ q\neq 0 \}$ when queue is non-empty and $\cog \leq \cog^*$.

\begin{figure*}
    \centering
	\begin{subfigure}[b]{0.21\textwidth}
	    \centering
        \includegraphics[width=1\linewidth, height=1\linewidth, keepaspectratio]{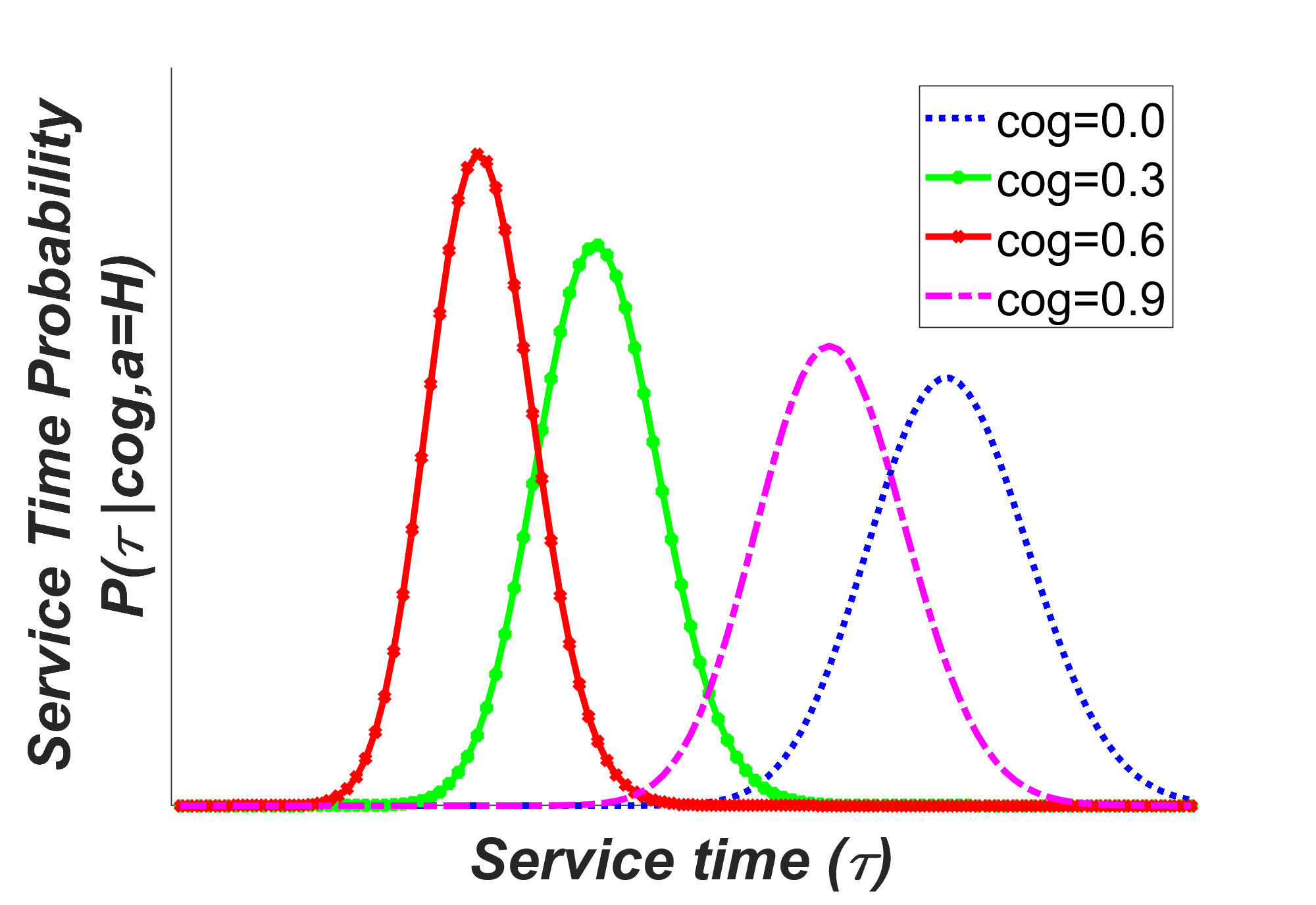}
        \caption{}
        \label{fig:sojurn_cog}
    \end{subfigure}
	\begin{subfigure}[b]{0.21\textwidth}
	    \centering
        \includegraphics[width=1\linewidth, height=1\linewidth, keepaspectratio]{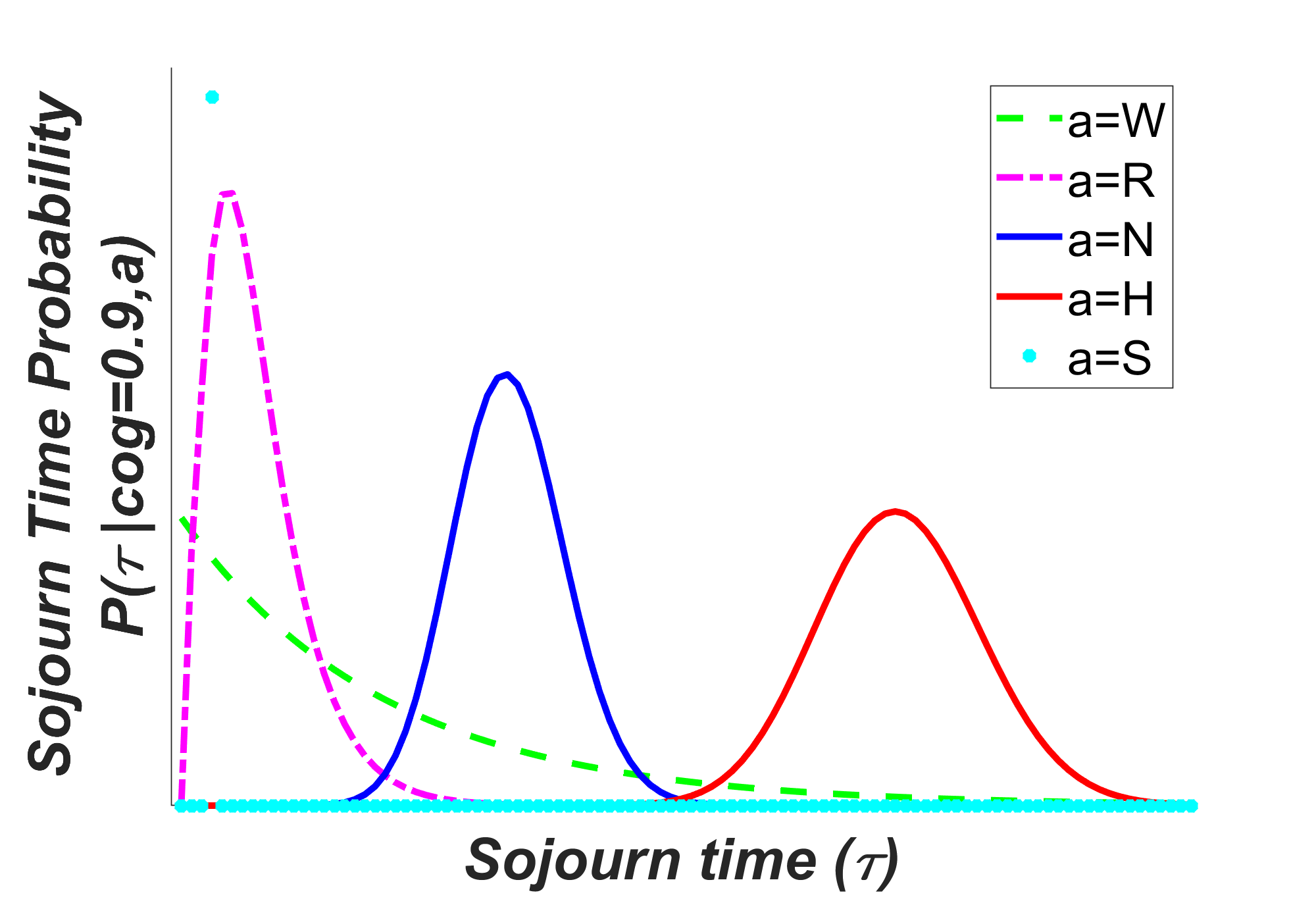}
        \caption{}
        \label{fig:sojurn_att}
    \end{subfigure}
    \begin{subfigure}[b]{0.21\textwidth}
        \centering
        \includegraphics[width=1\linewidth, height=1\linewidth, keepaspectratio]{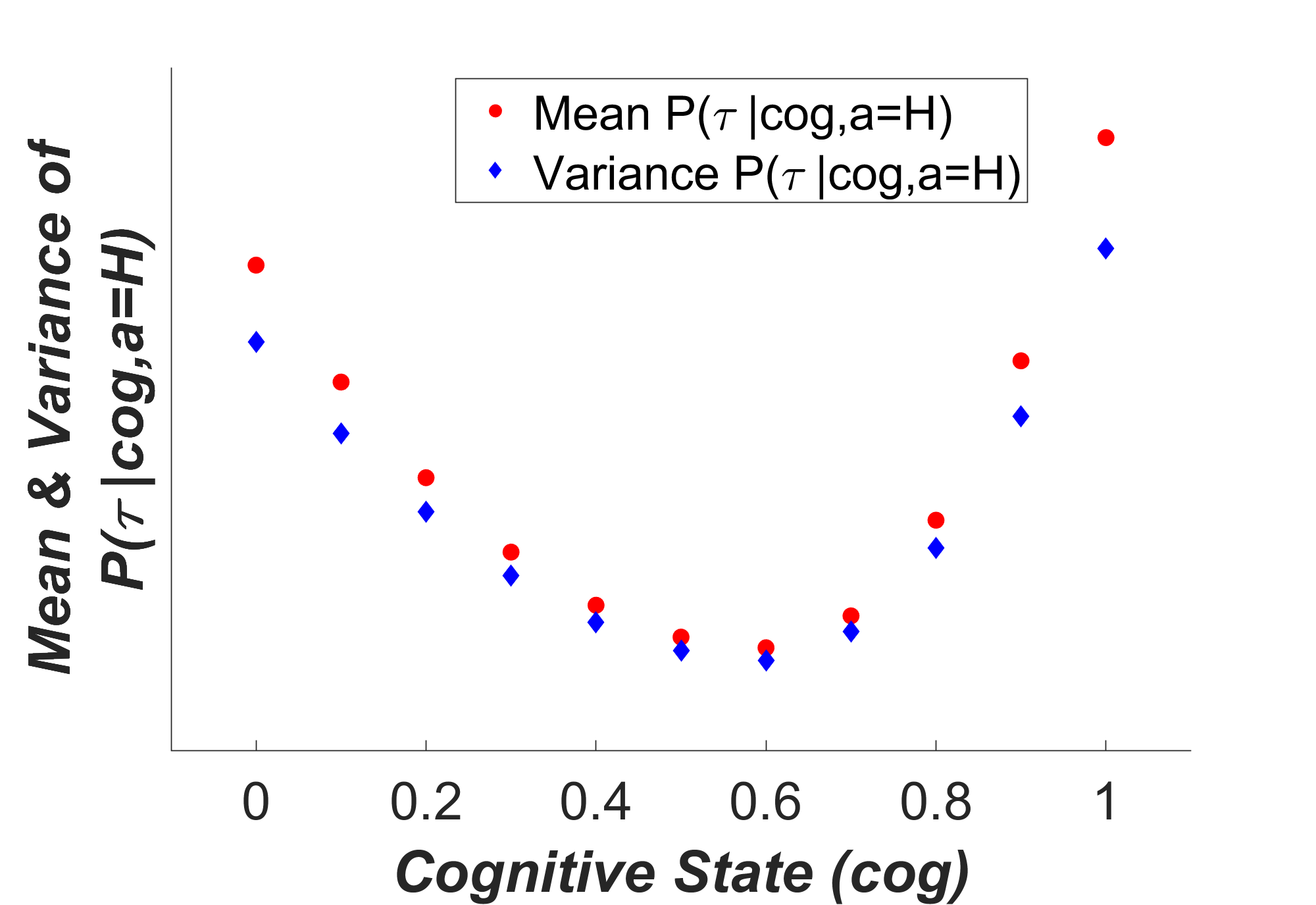}
        \caption{}
        \label{fig:mean_sojurn_cog}
    \end{subfigure}
	\begin{subfigure}[b]{0.21\textwidth}
	    \centering
        \includegraphics[width=1\linewidth, height=1\linewidth, keepaspectratio]{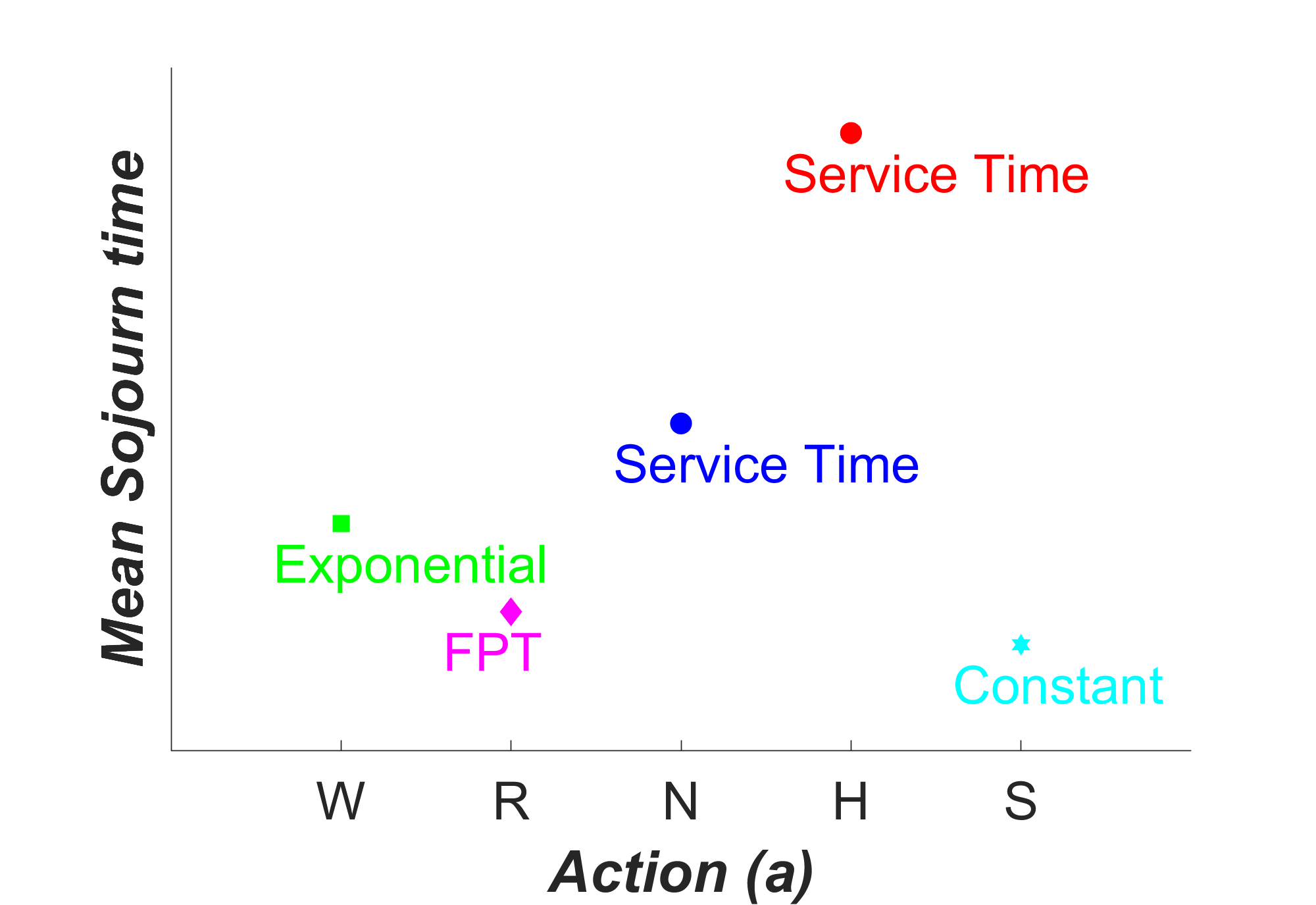}
        \caption{}
        \label{fig:mean_sojurn_att}
    \end{subfigure}
    \caption{\footnotesize{Service time distribution of the human operator with (a) varying cognitive state and high fidelity, (b) varying action and fixed cognitive state, $\cog=0.9$. 
    (c) Mean and variance of the service time distribution are unimodal functions of the cognitive state. (d) The mean sojourn time distribution takes on different forms based on the selected action.}}
    \label{fig:soj}
\end{figure*}
\item[(iii)] 
A state transition distribution  $\mathbb{P}\left(s'|\ \tau, s, a\right)$ from state $s$ to $s'$ for each action $a \in A_s$ conditioned on the discrete sojourn time $\tau \in \mathbb{R}_{>0}$ (time spent in state $s$ before transitioning into next state $s'$). The state transition from $s=(q, \ \cog) \to s'= \left(q', \ \cog'\right)$ consists of two independent transition processes which are given by (i) a Poisson process for transition from $q \to q'$ (ii) human cognitive dynamics for the transition from $\cog \to \cog'$. We model the cognitive dynamics of the human operator as a Markov chain in which, while servicing the task, the probability of an increase in cognitive state in small time $\delta t \in \mathbb{R}_{>0}$ is greater than the probability of a decrease in cognitive state. Furthermore, the probability of the increase in the cognitive state increases with the level of fidelity selected for servicing the task. Similarly, while waiting or resting, the probability of a decrease in cognitive state in small time $\delta t$ is higher than the probability of an increase in cognitive state. Sample parameters of the model used in our numerical simulations are shown in Table~\ref{tab:tab1}. This model of cognitive state dynamics is a stochastic equivalent of deterministic models of the utilization ratio considered in~\cite{KS-EF:10b}. It is assumed that the cognitive state remains unchanged when the human operator chooses to skip the task.
\begin{table}[htbp]
\caption{Cognitive Dynamics modeled as Markov chain}
\begin{center}
\begin{adjustbox}{width=0.42\textwidth}
\begin{tabular}{|c|c|c|c|}
\hline
& \textbf{\textit{Forward}} & \textbf{\textit{Backward}} & \textbf{\textit{Stay Probability$^{\mathrm{c}}$}} \\
%\cline{2-4} 
\textbf{Action} & \textbf{\textit{Probability$^{\mathrm{a}}$ ($\lambda_f \delta t$)}}& \textbf{\textit{Probability$^{\mathrm{b}}$ ($\lambda_b \delta t$)}}& \textbf{\textit{(1-$\lambda_f \delta t - \lambda_b \delta t$)}} \\
\hline
W& $\lambda_f=0.02$ (Noise) & $\lambda_b = 0.5$ & $1-0.52\delta t$ \\

\hline
R& $\lambda_f=0.02$ (Noise) & $\lambda_b = 0.5$ & $1-0.52\delta t$ \\

\hline
N& $\lambda_f=0.6$ & $\lambda_b = 0.02$ (Noise) & $1-0.62\delta t$ \\

\hline
H& $\lambda_f=1.1$ & $\lambda_b = 0.02$ (Noise) & $1-1.12\delta t$ \\

\hline
S& $\lambda_f=0$ & $\lambda_b = 0$ & $1$ \\
\hline
\multicolumn{4}{l}{$^{\mathrm{a}}$Forward Probability does not exist for $\cog=1$ (reflective boundary)} \\
\multicolumn{4}{l}{$^{\mathrm{b}}$Backward Probability does not exist for $\cog=0$ (reflective boundary)} \\
\multicolumn{4}{l}{$^{\mathrm{c}}$Stay Probability is $1-\lambda_f \delta t$ for $\cog=0$ and $1-\lambda_b \delta t$ for $\cog=1$}
\end{tabular}
\label{tab:tab1}
\end{adjustbox}
\end{center}
\end{table}
\item[(iv)] Sojourn time distribution $\mathbb{P}\left(\tau|\ s, a\right)$ of (discrete) time $\tau \in \mathbb{R}_{>0}$ spent in state $s$ until the next action is chosen takes on different forms depending on the selected action (Fig.~\ref{fig:mean_sojurn_att}). The sojourn time is the service time while servicing the task (normal/ high fidelity), resting time while resting, constant time of skipping $t_s \in \mathbb{R}_{>0}$ while skipping, and time until the next task arrival while waiting in case of an empty queue. We model the rest time as the time required to reach from the current cognitive state to the optimal cognitive state $\cog^*$. In our numerical illustrations, we model the service time distribution while servicing the task using a hypergeometric distribution (Fig.~\ref{fig:sojurn_cog} and~\ref{fig:sojurn_att}), where the parameters of the distribution are chosen such that the mean service time has the desired characteristics, i.e., it increases with the fidelity level (Fig.~\ref{fig:mean_sojurn_att}) and is a unimodal function of the cognitive state (Fig.~\ref{fig:mean_sojurn_cog}). While resting, sojourn time distribution is the first passage time (FPT) distribution for transitioning from the current cognitive state $\cog$ to $\cog^*$. We determine this distribution using matrix methods~\cite{diederich2003simple} applied to the Markov chain used to model the cognitive dynamics. Finally, to ensure the stability of the queue, we assume that the constant time of skip is less than $\frac{1}{\lambda}$, i.e., queue length decreases on average while skipping tasks.
\item[(v)] For selecting action $a$ at state $s$, the human receives a bounded reward $r(s, a)$ defined in~\eqref{eq:im-rew}. Additionally, the human incurs a penalty at a constant cost rate of $c$ due to each task waiting in the queue, and consequently, 
the cumulative expected cost for choosing action $a$ at state $s=(q,\cog)$ is given by:
    \begin{multline*}
    \sum_{\tau}\mathbb{P}(\tau|s,a)c\tau\left (\E\left [\frac{q+q'}{2}\Bigg|\ \tau, s,a\right ]\right )  \\
    = \sum_{\tau}\mathbb{P}(\tau|s,a)c\tau\left (\frac{2q+\lambda\tau}{2}\right ),
    \end{multline*}
which is obtained by using $\mathbb{E}[q|\tau,s,a] = q$ and $\mathbb{E}[q'|\tau,s,a] = q + \lambda\tau$. The  expected net immediate reward received by the operator for selecting an action $a$ in state $s$ is given by: 
\begin{align}
\!\!\!\!\!\! R(s,a) &= r(s,a) - \sum_{\tau}\mathbb{P}(\tau|s,a)c \left(\frac{2q+\lambda\tau}{2} \right) \tau \nonumber \\
&= r(s,a) -c\E\left[\tau|s,a\right]q - \frac{c\lambda}{2}\E\left[\tau^2|s,a\right],
\label{eq:1}
\end{align}
where $\E\left[\tau|\ s,a\right]$ and $\E\left[\tau^2|s,a\right]$ represent the first and the second conditional moment of the sojourn time distribution, respectively.
\item[(vi)] A discount factor $\gamma \in [0, 1)$, which we choose as $0.96$ for our numerical illustration.
\end{enumerate}
\begin{remark}
Although we assume a finite skip time, an alternative approach is to incorporate a penalty for the skip action. Note that, unlike a fixed penalty, a finite skip time results in a penalty that increases with queue length (see~\eqref{eq:1}). Consequently, the current approach is less inclined to skip tasks as the queue length increases compared to a model with a constant penalty.
\end{remark}
\begin{remark}
The reward $R(s, a)$ formulation can be interpreted as an unconstrained SMDP corresponding to a constrained SMDP that maximizes $r(s, a)$ subject to a constraint on the average queue length for the stability of the queue. Therefore, the penalty rate $c$ acts as the Lagrange multiplier for the unconstrained problem, and hence, can be obtained by primal-dual methods that use dual ascent for finding the Lagrange multiplier~\cite{agarwal2008structural}.
\end{remark}
\subsection{Solving SMDP for Optimal Policy}
For SMDP $\Gamma$,  the optimal value function $V^*: \mc{S} \rightarrow \mathbb{R}$ satisfies the following Bellman equation \cite{barto2003recent}:\\
\begin{equation}
V^{*}(s)=\max_{a\in{A_{s}}}\left[R(s,a)\ + \sum_{s',\tau}\gamma^{\tau}\mathbb{P}\left(s',\tau| s,a\right)V^{*}\left(s'\right)\right],
\label{eq:2}
\end{equation}
where $\mathbb{P}\left(s',\tau|s,a\right)$, which is the joint probability that a transition from state $s$ to state $s'$ occurs after time $\tau$ when action $a$ is selected can be rewritten as: 
\begin{equation}
\mathbb{P}\left(s',\tau|s,a\right)= \mathbb{P}\left(s'|\tau,s,a\right)\mathbb{P}\left(\tau|s,a\right),
\label{eq:3}
\end{equation}
where  $\mathbb{P}\left(s'|\tau,s,a\right)$ and $\mathbb{P}\left(\tau|s,a\right)$ are given by the state transition probability distribution and the sojourn time probability distribution, respectively. An optimal policy $\pi^*: \mc{S} \rightarrow \mc A_s$ at each state $s$ selects an action that achieves the maximum in~\eqref{eq:2}. 
We utilize the value iteration algorithm~\cite{sutton2018reinforcement} to compute an optimal policy.
\section{Numerical Illustrations} \label{Numerical Illustration}
We now numerically illustrate the optimal value function and an optimal policy for SMDP $\Gamma$.\\
\begin{figure}
    \centering
	\begin{subfigure}[b]{0.21\textwidth}
	    \centering
        \includegraphics[width=1\linewidth, height=1\linewidth, keepaspectratio]{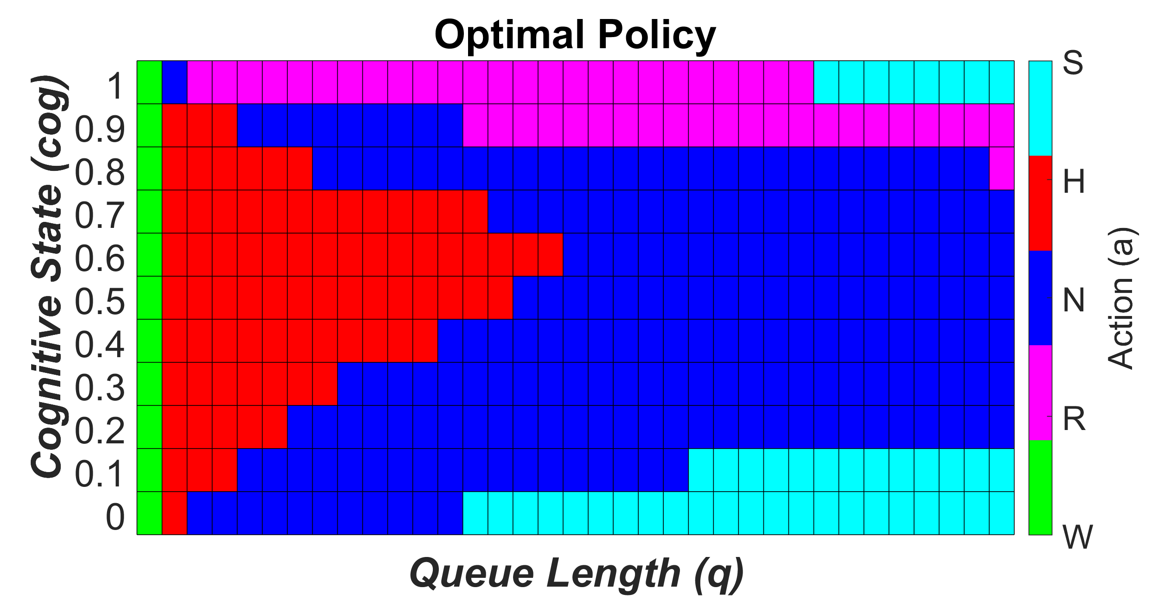}
        \caption{}
        \label{fig:Optimal_policy}
    \end{subfigure}
	\begin{subfigure}[b]{0.21\textwidth}
	    \centering
        \includegraphics[width=1\linewidth, height=1\linewidth, keepaspectratio]{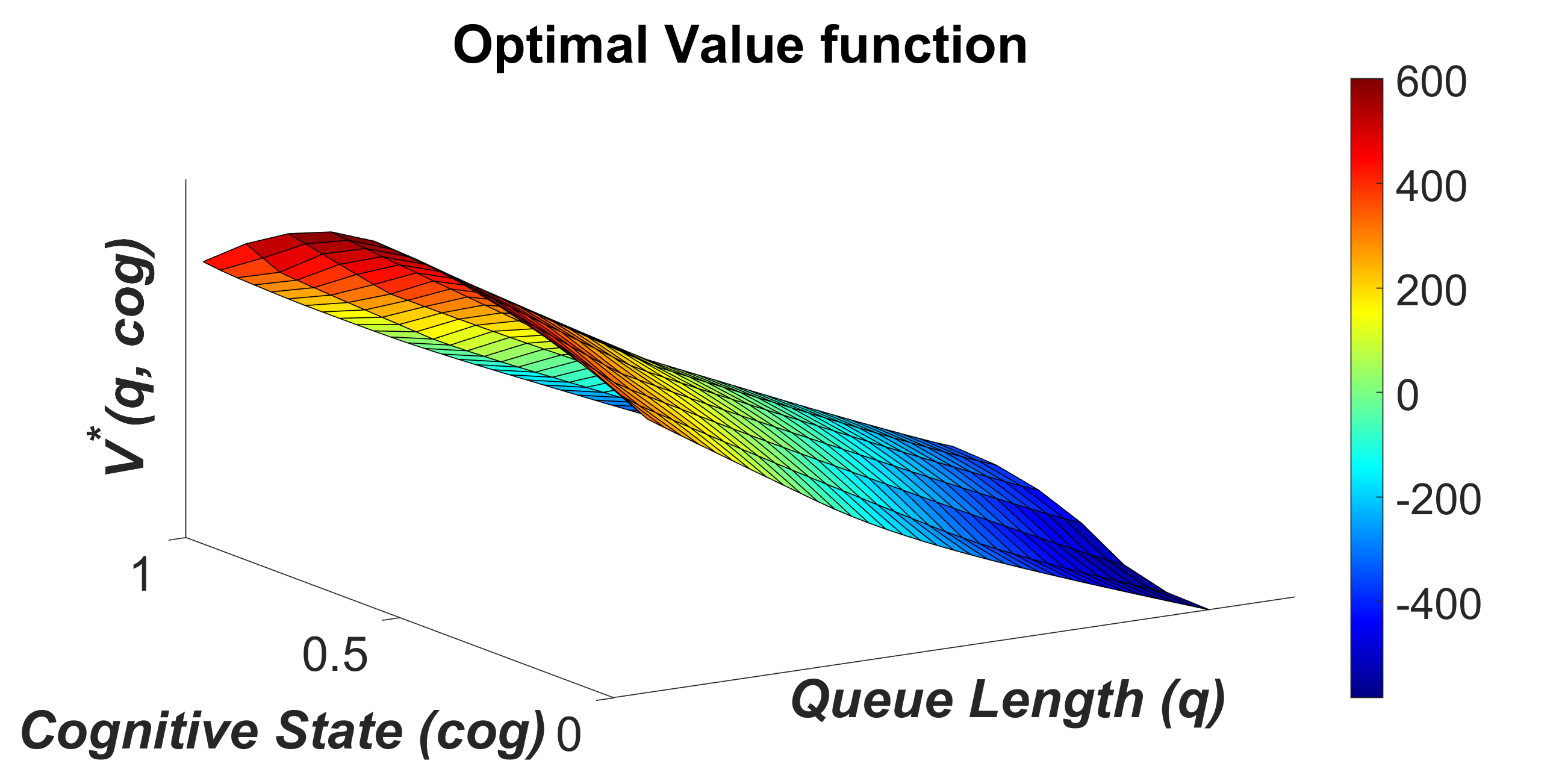}
        \caption{}
        \label{fig:Optimal_value}
    \end{subfigure}
    
    \caption{\footnotesize(a) Optimal Policy $\pi^*$ and (b) Optimal Value Function $V^*$ for SMDP $\Gamma$ where the time required to skip the tasks is not too small compared to the mean service time.}
    \label{fig:Optimal}
\end{figure}
\begin{figure}[ht]
    \centering
	\begin{subfigure}[b]{0.15\textwidth}
	    \centering
        \includegraphics[width=1\linewidth, height=1\linewidth, keepaspectratio]{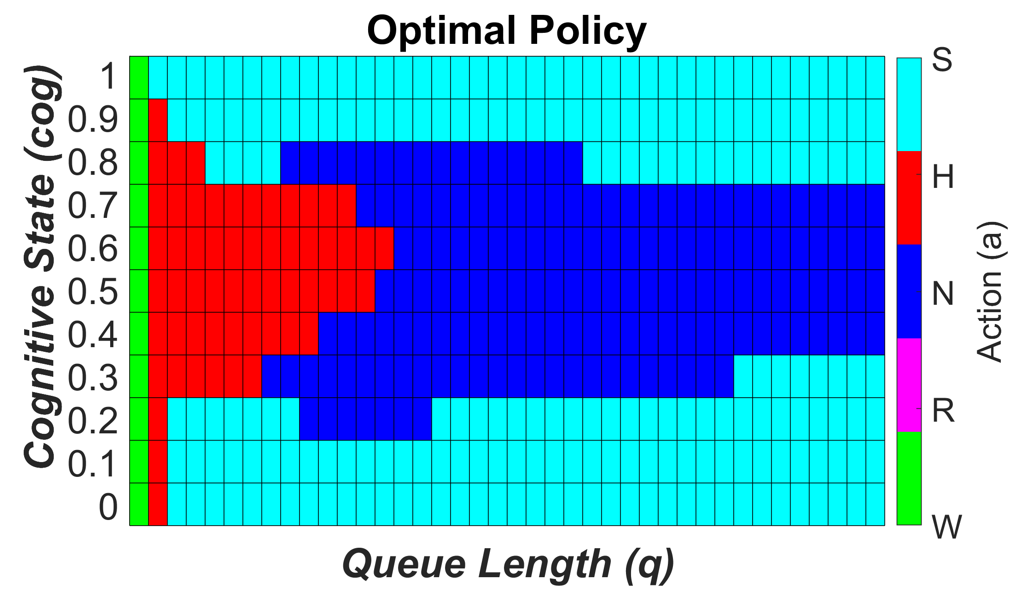}
        \caption{}
        \label{fig:Optimal_policy_05_}
    \end{subfigure}
    \begin{subfigure}[b]{0.15\textwidth}
	    \centering
        \includegraphics[width=1\linewidth, height=1\linewidth, keepaspectratio]{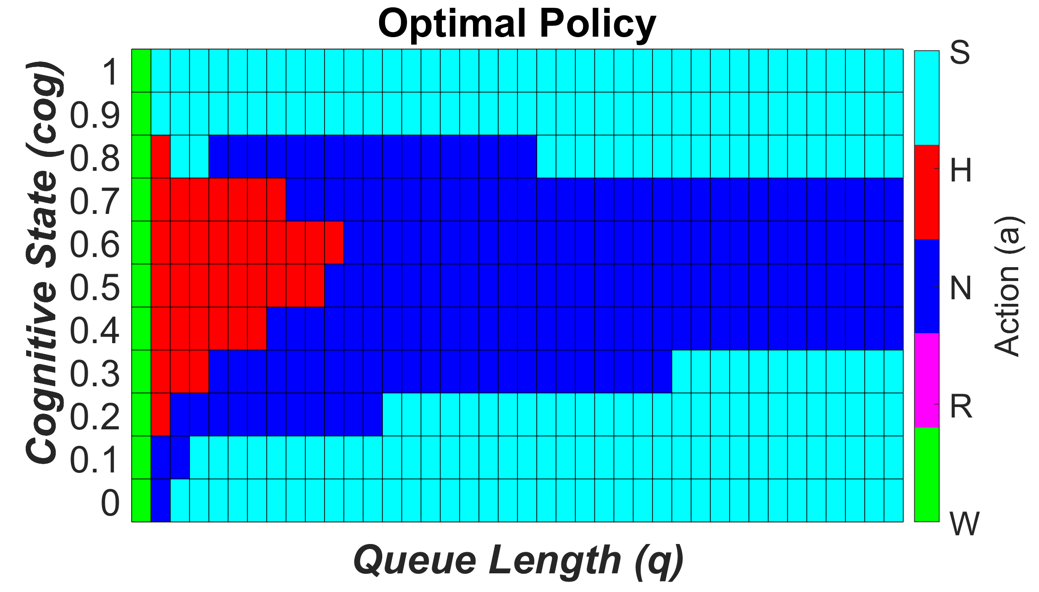}
        \caption{}
        \label{fig:Optimal_policy_1_}
    \end{subfigure}
	\begin{subfigure}[b]{0.15\textwidth}
	    \centering
        \includegraphics[width=1.1\linewidth, height=2.2\linewidth, keepaspectratio]{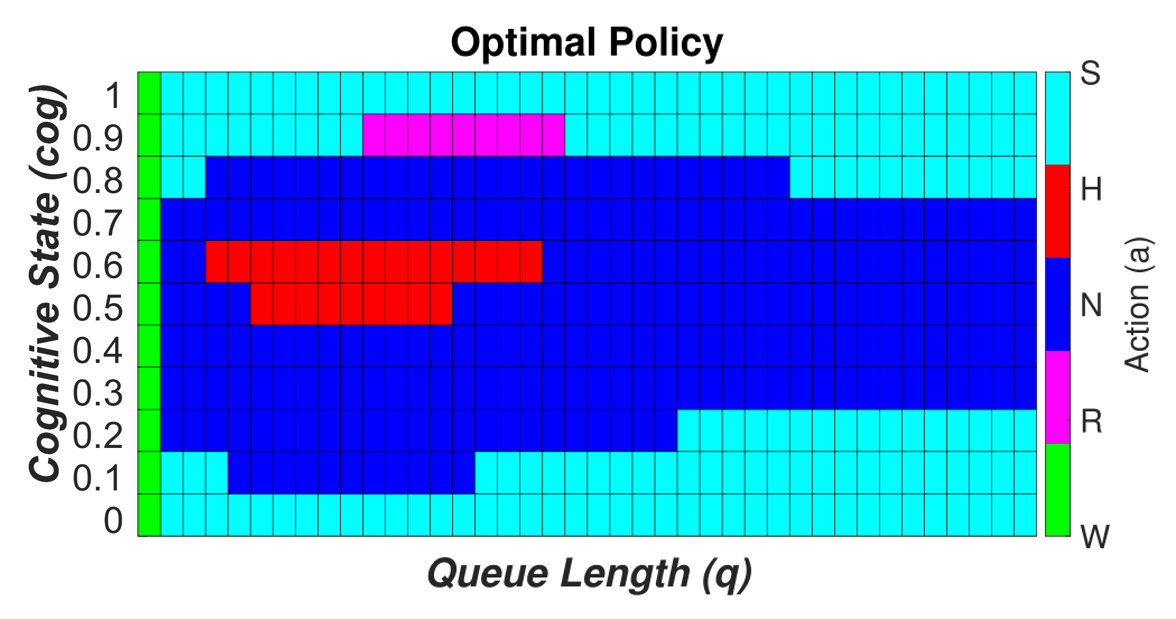}
        \caption{}
        \label{fig:Optimal_value_15_}
    \end{subfigure}
    
    \caption{\footnotesize{Optimal policy $\pi^*$ for $(a) \lambda = 0.5$, (b) $\lambda = 1$, and (c) $\lambda = 4$.} In cases (a) and (b), action \textit{S} in the optimal policy does not have a unique threshold for some cognitive states. Similarly, in case (c), actions \textit{S} and \textit{N} in the optimal policy do not have unique thresholds.}
    \label{fig:Optimal_skip_}
\end{figure}
Fig.~\ref{fig:Optimal_policy} and~\ref{fig:Optimal_value} show an optimal policy $\pi^{*}$, and the optimal value function $V^{*}$, respectively, for the case in which the skip time is not too small compared to the mean service time. If the skip time is too small, the action $S$ is the optimal action almost everywhere to reduce the queue length. For a sufficiently high arrival rate $\lambda$ such that there is always a task in the queue after servicing the current task, we observe that for any given $\cog$, $V^*$ is monotonically decreasing with $q$.\\
Additionally, we observe that for a given $q$, $V^*$ is an unimodal function of $\cog$, with its maximum value corresponding to the optimal cognitive state ($\cog^* =0.6$ for numerical illustrations). We observe that $\pi^*$ selects the high fidelity level around the $\cog^*$ for low queue length, and thereafter transitions to a normal fidelity level for higher queue lengths. We also observe that in low cognitive states, the optimal policy is to keep skipping the tasks until the queue length becomes small, and then start servicing the tasks. In higher cognitive states, we observe that resting is the optimal action at smaller queue lengths while skipping tasks is the optimal choice at larger queue lengths. Additionally, we observe the effect of $\cog$ on $\pi^*$. In particular, we observe that $\pi^*$ switches from \textit{H} to \textit{N}, \textit{N} to \textit{R}, and \textit{R} to \textit{S} at certain thresholds on $q$, and these thresholds appear to be a unimodal function of $\cog$. This behavior can be attributed to the mean service time being unimodal w.r.t $\cog$. \\
Fig.~\ref{fig:Optimal_skip_} shows some examples of $\pi^*$ for certain parameters. We observe that for some cognitive states, $\pi^*$ does not have a unique threshold and the same action reappears after switching to another action. For example, in Fig.~\ref{fig:Optimal_policy_05_}, action \textit{S} is observed between actions \textit{H} and \textit{N}, as well as after action \textit{N}. In the following section, we provide sufficient conditions under which  $\pi^*$  has unique transition thresholds at which actions switch, and the previous action does not re-appear for the same $\cog$.

\section{Structural Properties of the Optimal Policy} \label{Structural Properties}
% In this section, we characterize the structural properties of the optimal policy $\pi^*$ for SMDP $\Gamma$. We provide sufficient conditions under which, for each cognitive state, the optimal fidelity selection policy has unique transition thresholds after which the optimal policy switches actions, and the previous action does not re-appear for the same cognitive state. \\
We establish the structural properties of the optimal infinite-horizon value function by considering the finite horizon case and then extending the results to the infinite horizon by taking the infinite step limit.\\
Let $V_n^*(s_0), n \ge 0,$ be the discounted $n$-step optimal expected reward when the initial state is $s_0$, where $V_0^*(q,\cog)= -Cq$ {is the terminal cost for the finite-horizon case for a non-negative constant $C$}. Each step size $k \in \{0, \ldots, n-1 \}$ is based on the sojourn time $\tau_k$, spent in a state $s_k$ when action $a_k$ is selected. Let $J_{n,\pi}(s_0)$ denote the discounted $n$-step expected reward with initial state $s_0$ under a given policy $\pi$. Henceforth, for brevity of notation, we denote the conditional expectation $\E[\cdot | s_0, \pi]$ by $\E_{\pi}[\cdot]$. $J_{n,\pi}(s_0)$ is given by:\\
\begin{equation}{\label{eq:60}}
 \!\!\!\!\!\! J_{n,\pi}(s_0)=\E_{\pi}\left[\sum_{i=0}^{n-1}\gamma^{\zeta_{i}}R(s_i,a_i) - \gamma^{\zeta_{n}}Cq_n \ \right],
\end{equation}
where $\zeta_{i}:={\sum_{j=0}^{i-1}\tau_j}$ for $i>0$ and $\zeta_{0}:=0$.
The discounted $n$-step optimal expected reward $V_n^*(s)$ is given by: 
\begin{equation}{\label{eq:61}}
       V_n^*(s_0)= J_{n,\pi^*}(s_0),
\end{equation}
where $\pi^*$ is the optimal policy that maximizes $J_{n,\pi}(s_0)$ at each $s_0$. \\
% \st{We introduce the following notation. Let $\map{q_j^*}{\mc C}{\integer_{\ge 0} \union \{+\infty\}} $, for $j \in \{1,2,3\}$ be some functions of the cognitive state.} 
Let $\map{\mu^1}{\mc S\times \mc A_s}{\real_{>0}}$ and $\map{\mu^2}{\mc S\times \mc A_s}{\real_{>0}}$ be function defined by $\mu^1(s,a)= \E\left[\tau | s,a\right]$ and $\mu^2(s, a)= \E\left[\tau^2|s,a\right]$, where $\tau$ is the sojourn time. {We study the structural properties of the optimal policy for a large queue capacity, i.e. in the limit $L \to +\infty$,} 
% \st{We study the properties of SMDP $\Gamma$ in the limit $L \to +\infty$,} 
and under the following assumptions:
\begin{enumerate}
    \item [(A1)] {The task arrival rate $\lambda$ is sufficiently high so that the queue is never empty with high probability.}
        \item [(A2)] {For any state $s=(q, \cog)$\footnote{The action $R$ is only available for states with $\cog>\cog^*$.}}:
\begin{align}
\begin{split}
 &\!\!\!\!\!\!\!\! \mu^1(s, S)< \mu^1(s, R) < \mu^1(s, N) < \mu^1(s, H), \text{ and }  \\
&\!\!\!\!\!\!\!\! \mu^2(s, S)< \mu^2(s, R) < \mu^2(s, N) < \mu^2(s, H). 
\label{eq:7}
\end{split}
\end{align}
\item[(A3)] {We assume that $\E_{\pi}[\gamma^{\tau}] \le f(\E_{\pi}[\tau], \var_{\pi}(\tau) ) < 1$, where $\var_{\pi}(\tau) = \var(\tau|s_0, a=\pi(s_0))$ is the variance of $\tau$ in any initial state $s_0$ under a given policy $\pi$, and $f$ is a monotonic function such that $f(\cdot, \var_{\pi}(\tau))$ is monotonically decreasing and  $f(\E_{\pi}[\tau] ,\cdot)$ is monotonically increasing.}   
\end{enumerate}
We make the assumption (A1) for convenience. Indeed, if the queue is allowed to be empty, then we will need to deal with an extra ``waiting" action. Also, high arrival rates are the most interesting setting to study optimal fidelity selection. Assumption (A2) is true for a broad range of interesting parameters that define sojourn time distribution(s). {Assumption (A3) holds for a class of light-tail distributions with non-negative support for $\tau$}, for example, let the  moment generating function (MGF) of $\tau$ be upper bounded by the MGF of Gamma distribution. Therefore, we have
% \[\E[\gamma^{\zeta}] \le \gamma^{\E[\zeta]}e^{\frac{\var(\zeta)(\ln(\gamma))^2}{2}} =: f(\E[\zeta], \var(\zeta)).\] 
\begin{align*}
 \E_{\pi}[e^{t\tau}] &\le  \left(1-\frac{\var_{\pi}(\tau) t}{\E_{\pi}[\tau]}\right)^{\frac{-\E_{\pi}[\tau]^2}{\var_{\pi}(\tau)}}, \text{for all $t < \frac{\E_{\pi}[\tau]}{\var_{\pi}(\tau)}$.} 
 \end{align*}
 Substituting $t=\ln(\gamma) <0 <  \frac{\E_{\pi}[\tau]}{\var_{\pi}(\tau)}$, we get
 \begin{align*}
 \E_{\pi}[\gamma^{\tau}] &\le \left(1-\frac{\var_{\pi}(\tau) \ln(\gamma)}{\E_{\pi}[\tau]}\right)^{\frac{-\E_{\pi}[\tau]^2}{\var_{\pi}(\tau)}} \\
 &=: f(\E_{\pi}[\tau], \var_{\pi}(\tau)).
 \end{align*}
Let $\rho := \max_{\cog,a}f(\E[\tau|\cog,a], \var(\tau|\cog,a))$. Therefore, $\E_{\pi}[\gamma^{\tau}] \le \rho$. For the class of distributions of $\tau$ satisfying assumption (A3), and  any initial state $s_0$ and policy $\pi$, we have
\begin{align*}
    \E_{\pi}[\gamma^{\zeta_k}] &\overset{(1)^*}{=}  \prod_{i=0}^{k-1} \E_{\pi}[\gamma^{\tau_i}] \le \prod_{i=0}^{k-1} f(\E_{\pi}[\tau_i], \var_{\pi}(\tau_i)) \le \rho^k,  
\end{align*}
 where $(1)^*$ follows from the independence of $\tau_i$ and $\tau_j$, for $i\neq j$. 
 Therefore, we have
 \begin{align*}
     \lim_{n \rightarrow \infty}\sum_{k=0}^{n-1}\E_{\pi}[\gamma^{\zeta_k}] \le    \lim_{n \rightarrow \infty} \sum_{k=0}^{n-1} \rho^k = \frac{1}{1-\rho}.
 \end{align*}
We will now establish that the optimal policy for SMDP $\Gamma$ is a threshold-based policy if the following condition holds {for each cognitive state $\cog$}:
\begin{multline}\label{eq:9}
\min\{\mathbb{E}[\tau|\cog,H]-\mathbb{E}[\tau|\cog,N],  \ \mathbb{E}[\tau|\cog,N]- \\ \mathbb{E}[\tau|\cog,R], \ \mathbb{E}[\tau|\cog,R]-t_s \} + \\ \frac{t_s\gamma^{\E[\tau| \cog, H]}}{1-\gamma^{t_{\max}}}  \ge \frac{t_{\max}}{1 - \rho} \max_{a \in\mc A_{s}}\E[\gamma^{\tau}| \cog, a],
\end{multline}
 where 
\begin{equation*}
\rho = \max_{\cog,a}f(\E[\tau|\cog,a], \var(\tau|\cog,a))
\end{equation*} 
 is an upper bound for $\E_{\pi}[\gamma^{\tau}]$, $t_{\max} =\E[\tau | \cog=1, a=H] $ is the maximum expected sojourn time {(assuming largest mean service time in highest cognitive state)}, and $t_s$ is the constant time for skip.\\
 {\begin{remark}
For tasks with large differences in expected sojourn times, i.e., $0 \ll t_s  \ll \E[\tau | \cog, R]  \ll \E[\tau| \cog, N] \ll \E[\tau| \cog, H]$, $\max_{a \in \mc A_{s}}\E[\gamma^{\tau}| \cog , a] \rightarrow 0$, and~\eqref{eq:9} always holds.
 \end{remark}}
 {We introduce the following notation. Let $\map{q_j^*}{\mc C}{\integer_{\ge 0} \union \{+\infty\}} $, for $j \in \{1,2,3\}$ be some functions of the cognitive state.}
\begin{theorem}[\bit{Structure of optimal policy}]\label{thm:thm1}
For SMDP $\Gamma$ under assumptions {(A1-A3)} and an associated optimal policy $\pi^*$, if 
the difference in the expected sojourn times is sufficiently large such that~\eqref{eq:9} holds for any cognitive state $\cog$, then the following statements hold:
\begin{enumerate}
    \item[(i)] there exists unique threshold functions $q_1^*(\cog)$, $q_2^*(\cog)$, and $q_3^*(\cog)$
    such that for each $\cog > \cog^*$:
\begin{equation}
\!\!\!\!\!\!\!\!\pi^*(s=(q,\cog))=\begin{cases}
\textit{H},  &  q \leq q_1^*(\cog), \\
                  \textit{N},  &  q_1^*(\cog) <  q \leq q_2^*(\cog), \\
                  \textit{R}, & q_2^*(\cog)< q \leq q_3^*(\cog), \\
                  \textit{S},  &  q > q_3^*(\cog); \\
\end{cases}\label{eq:theorem}
\end{equation}
\item[(ii)] there exists unique threshold functions $q_1^*(\cog)$ and $q_2^*(\cog)$ such that for any $\cog \leq \cog^*$:
\begin{equation}
\!\!\!\!\!\!\!\!\pi^*(s=(q,\cog))= \begin{cases}
                  \textit{H},  & q \leq q_1^*(\cog), \\
                  \textit{N},  &  q_1^*(\cog) <  q \leq q_2^*(\cog), \\
                  \textit{S}, & q > q_2^*(\cog).
\end{cases} 
\label{eq:theorem2}
\end{equation}
\end{enumerate}
\end{theorem}
We prove Theorem~\ref{thm:thm1} with the help of the following lemmas.\\
\begin{lemma}\label{lemma:reward}(\bit{Immediate Reward}):
For SMDP $\Gamma$, the immediate expected reward $R(s,a)$, for each $a \in A_s$
\begin{enumerate}
    \item[(i)]  is linearly decreasing with queue length $q$ for any fixed cognitive state $\cog$; 
    \item[(ii)] is a unimodal function\footnote{The expected immediate reward under action \textit{S} is a constant, which we treat as a unimodal function.} of the cognitive state $\cog$ for any fixed queue length $q$ with its maximum value achieved at the optimal cognitive state $\cog^*$.
\end{enumerate}
\end{lemma}
\begin{proof}
See Appendix~\ref{app_lemma_reward} for the proof.
\end{proof}
We now provide important mathematical results in Lemma~\ref{app:1} which we use to establish Lemma~\ref{lemma:value_function_bound}.
\begin{lemma}\label{app:1}
For the SMDP $\Gamma$, the following equations hold for any initial state $s_0$ and policy $\pi$:
\begin{align*}
  (i) \E_{{\pi}}\left[\sum_{k=0}^{n-1}\gamma^{\zeta_k}\E[\tau_k^2|\cog_k,{a}_k] \right]= \sum_{k=0}^{n-1}\E_{\pi} \left[ \gamma^{\zeta_k}\tau_k^2\right];
\end{align*}
\begin{multline}
(ii) \E_{{\pi}}\left[\sum_{k=0}^{n-1}\gamma^{\zeta_k}\E[\tau_k|\cog_k,{a}_k]{q}_k  \right]=   \nonumber \\  \sum_{k=0}^{n-1}\E_{\pi} \left[ \gamma^{\zeta_k}\tau_k \E_{\pi}\left[{q}_k| s_0, \zeta_k\right] \right].
\end{multline}
\end{lemma}
\begin{proof}
See Appendix~\ref{app_app} for the proof.
\end{proof}
\begin{lemma}\label{lemma:value_function_bound}(\bit{Value function bounds}):
For SMDP $\Gamma$ under assumptions {(A1-A3)}, for any $\tilde{q}_0 \ge q_0$,
 $0 \le \frac{ct_s\Delta q}{1-\gamma^{t_{\max}}} \le V^*(q_0,\cog_0)-V^*(\tilde{q}_0,\cog_0) \le {\frac{ct_{\max}\Delta q}{1 - \rho}}$, where $\Delta q = \tilde{q}_0-q_0$, {$\rho$ is an upper bound on $\E_{\pi}[\gamma^{\tau}]$,} $t_{\max} =\E[\tau | \cog=1, a=H] $ is the maximum expected sojourn time, and $t_s$ is the constant time for skip.
\end{lemma}
\begin{proof}
See Appendix~\ref{app_value_function_bound} for the proof.
\end{proof}
\begin{remark}
It follows from Lemma~\ref{lemma:value_function_bound}, that for SMDP $\Gamma$ under assumptions {(A1-A3)}, the optimal value function $V^*(q, \cdot)$ is monotonically decreasing with queue length $q$.
\end{remark}
 \begin{lemma}\label{lemma:threshold}(\bit{Thresholds for low cognitive states}):
For the SMDP $\Gamma$ under assumptions (A1-A3), and an associated optimal policy $\pi^*$, the following statements hold for each $\cog\le \cog^*$: 
\begin{enumerate}
    \item[(i)]  there exists a threshold function $q_1^*(\cog)$, such that the action \textit{N} strictly dominates action \textit{H}, for each $q > q_1^*(\cog)$ if
\begin{multline*}
      \mathbb{E}[\tau|\cog,H]-\mathbb{E}[\tau|\cog,N] + \frac{t_s\gamma^{\E[\tau|\cog,H]}}{1-\gamma^{t_{\max}}}  \\
    \ge \frac{t_{\max}}{1 - \rho}\E[\gamma^{\tau}| \cog , N];
\end{multline*} 
\item[(ii)] there exists a threshold function $q_2^*(\cog)$, such that for each $q > q_2^*(\cog)$, action \textit{S} is optimal if 
\begin{multline*}
\mathbb{E}[\tau|\cog,N]-t_s + \frac{t_s\gamma^{\E[\tau|\cog,H]}}{1-\gamma^{t_{\max}}}  \ge \gamma^{t_s} \frac{t_{\max}}{1 - \rho}.
\end{multline*}
\end{enumerate}
\end{lemma}
\begin{proof}
See Appendix~\ref{app_threshold1} for the proof.
\end{proof}
 \begin{lemma}\label{lemma:threshold_2}(\bit{Thresholds for high cognitive states}):
For the SMDP $\Gamma$ under assumptions (A1-A3), and an associated optimal policy $\pi^*$, the following statements hold for each $\cog> \cog^*$: \begin{enumerate}
       \item[(i)]  there exists a threshold function $q_1^*(\cog)$, such that the action \textit{N} strictly dominates action \textit{H}, for each $q > q_1^*(\cog)$ if
    \begin{multline*}
       \mathbb{E}[\tau|\cog,H]-\mathbb{E}[\tau|\cog,N] + \frac{t_s\gamma^{\E[\tau|\cog,H]}}{1-\gamma^{t_{\max}}} \\ \ge \frac{t_{\max}}{1 - \rho}\E[\gamma^{\tau}| \cog , N];
    \end{multline*}
    \item[(ii)] there exists a threshold function $q_2^*(\cog)$, such that the action \textit{R} strictly dominates actions \textit{H} and \textit{N}, for each $q > q_2^*(\cog)$ if
    \begin{multline*}
        \mathbb{E}[\tau|\cog,N]-\mathbb{E}[\tau|\cog,R] + \frac{t_s\gamma^{\E[\tau|\cog,H]}}{1-\gamma^{t_{\max}}} \\ \ge \frac{t_{\max}}{1 - \rho}\mathbb{E}[\gamma^\tau|\cog,R].
    \end{multline*}
\item[(iii)] there exists a threshold function $q_3^*(\cog)$, such that for each $q > q_3^*(\cog)$, action \textit{S} is optimal if 
\begin{multline*}
    \mathbb{E}[\tau|\cog,R]-t_s + \frac{t_s\gamma^{\E[\tau|\cog,H]}}{1-\gamma^{t_{\max}}}  \ge \gamma^{t_s}\frac{t_{\max}}{1 - \rho}.
\end{multline*}
\end{enumerate}
\end{lemma}
\begin{proof}
Recall that $\mc A_s := \{\{$\textit{R}, \textit{S}, \textit{N}, \textit{H} \}$| \ s \in \mc{S}, \ q\neq 0 \}$ when queue is non-empty and $\cog>\cog^*$. The proof of Lemma~\ref{lemma:threshold_2} follows analogously to the proof of Lemma~\ref{lemma:threshold}.
\end{proof}
\textit{Proof of Theorem~\ref{thm:thm1}:}
The proof follows by finding the intersection of the sufficient conditions from Lemmas~\ref{lemma:threshold} and~\ref{lemma:threshold_2} to get the condition: 
\begin{multline}{\label{eq:167newww}}
\min\{\mathbb{E}[\tau|\cog,H]-\mathbb{E}[\tau|\cog,N],  \ \mathbb{E}[\tau|\cog,N]- \\ \mathbb{E}[\tau|\cog,R], \ \mathbb{E}[\tau|\cog,R]-t_s\} + \\ \frac{t_s\gamma^{\E[\tau| \cog, H]}}{1-\gamma^{t_{\max}}}  \ge \frac{t_{\max}}{1 - \rho} \max_{a \in \mc A_{s}}\E[\gamma^{\tau}| \cog, a],
\end{multline}
under which the optimal policy $\pi^*$ satisfies Theorem~\ref{thm:thm1}.

\section{Conclusions and Future Directions} \label{Conclusions}
We studied optimal fidelity selection for a human operator servicing a stream of homogeneous tasks using an SMDP framework. In particular, we studied the influence of human cognitive dynamics on an optimal fidelity selection policy. We presented numerical illustrations of the optimal policy and established its structural properties. These structural properties can be leveraged to tune the design parameters, deal with the model uncertainty, or determine a minimally parameterized policy for specific individuals and tasks.\\ 
There are several possible avenues for future research. An interesting direction is to conduct experiments with human subjects, measure EEG signals to assess their cognitive state and test the benefits of recommending optimal fidelity levels. It is of interest to extend this work to a team of human operators servicing a stream of heterogeneous tasks. A preliminary setup is considered in~\cite{gupta2019achieving,gupta2022incentivizing}, where authors study a game-theoretic approach to incentivize collaboration in a team of heterogeneous agents. In such a setting, finding the optimal routing and scheduling strategies for these heterogeneous tasks is also of interest. 
\begin{ack}                               % Place acknowledgements
This work has been supported by NSF Award IIS-1734272 and ECCS-2024649. % here.
\end{ack}
\bibliographystyle{ieeetr}
\bibliography{autosam}           % and a bib file to produce the 

.

\appendix
\section{Proof of Lemma~\ref{lemma:reward} [\bit{Immediate Reward}]}\label{app_lemma_reward}
We start by establishing the first statement. 
Recall that the expected net immediate reward
% Total Expected Immediate Reward 
received by the human operator for selecting action $a$ in state $s$ is given by~\eqref{eq:1}.
% ~\eqref{eq:1} can be rewritten as:
%
% \begin{equation}
% R(s,a) = -c\E\left[\tau|s,a\right]q + r(s,a) - \frac{c\lambda}{2}\E\left[\tau^2|s,a\right].
% \label{eq:4}
% \end{equation}
% where $\E\left[\tau|\ s,a\right]$ and $\E\left[\tau^2|s,a\right]$ represents the first and the second moment of the sojourn time distribution, respectively. 
We note that the moments of the sojourn time distribution are independent of the queue length $q$. Therefore,~\eqref{eq:1} can be re-written as:
\begin{equation}
R(s,a) = -a_1(\cog,a)q + a_2(\cog, a),
\label{eq:5}
\end{equation}
 where $a_1 (\cog,a)=c\E\left[\tau|s,a\right]$ and $a_2(\cog,a)=r(s, a) - \frac{c\lambda}{2}\E\left[\tau^2|s,a\right]$.\\
 For a fixed cognitive state $\cog$ and action $a$, both $a_1$ and $a_2$ are constants and therefore, the expected immediate reward linearly decreases with the queue length $q$ and the first statement follows.\\
 The second statement follows by observing that, for a given queue length $q$, the mean and variance of the sojourn time for each action $a \in A_s$ are unimodal functions of the cognitive state with their peaks at  $\cog^*$ (Fig.~\ref{fig:mean_sojurn_cog}).

$\hfill \oprocend$
\section{Proof of Lemma~\ref{app:1}}\label{app_app} 
For brevity of notation, let $\sum_{a}\sum_{b}\sum_{c}\sum_{d}(\cdot)$ be denoted by $\sum_{\Lambda}(\cdot)$, where $\Lambda :=\{a,b,c,d\}$ for any arbitrary $a,b,c,d$. We start by establishing the first statement. Let $T := \E_{{\pi}}\left[\sum_{k=0}^{n-1}\gamma^{\zeta_k}\E[\tau_k^2|\cog_k,{a}_k] \right]$ and $\Delta := \{k, \zeta_k, \cog_k, \tau_k\}$. We have
\begin{align*}
T&=\sum_{k=0}^{n-1}\E_{{\pi}}\left[ \gamma^{\zeta_k}\E[\tau_k^2|\cog_k,{a}_k]  \right] \\
&=\sum_{k=0}^{n-1}\E_{{\pi}}\left[ \gamma^{\zeta_k}\sum_{\tau_k}\tau_k^2 \mathbb{P}(\tau_k|\cog_k, a_k)  \right] \\
% &=\sum_{k=0}^{n-1}\sum_{\zeta_k,\cog_k}\left\{ \gamma^{\zeta_k}\E[\tau_k^2|\cog_k,\pi] \mathbb{P}( \zeta_k, \cog_k | \pi) \right\} \\
&=\sum_{\Delta}\left\{ \gamma^{\zeta_k} \tau_k^2 \mathbb{P}(\tau_k|\cog_k,a_k)  \mathbb{P}( \zeta_k, \cog_k |s_0, \pi) \right\} \\
% &=\sum_{k=0}^{n-1}\sum_{\zeta_k,\cog_k}\left\{ \gamma^{\zeta_k} \sum_{\tau_k}\tau_k^2 \mathbb{P}(\tau_k|\cog_k,\zeta_k, s_0, \pi) \mathbb{P}( \zeta_k, \cog_k ) \right\} \\
&\overset{(2)^*}{=}\sum_{\Delta}\left\{ \gamma^{\zeta_k} \tau_k^2 \mathbb{P}(\tau_k, \cog_k , \zeta_k |s_0, \pi) \right\} \\
% &=\sum_{k=0}^{n-1}\sum_{\zeta_k}\left\{ \gamma^{\zeta_k} \sum_{\tau_k}\tau_k^2 \mathbb{P}(\tau_k, \zeta_k|s_0, \pi) \right\} \\
&=\sum_{k=0}^{n-1}\left\{ \sum_{\zeta_k} \sum_{\tau_k} \gamma^{\zeta_k}\tau_k^2 \mathbb{P}(\tau_k , \zeta_k|s_0, \pi)\right\} \\
&= \sum_{k=0}^{n-1}\E_{\pi} \left[ \gamma^{\zeta_k}\tau_k^2 \right],
\end{align*}
{where $(2)^*$ follows by using $\pi(s_k) = a_k$}.\\
We now establish the second statement. Let $G:= \E_{{\pi}}\left[\sum_{k=0}^{n-1}\gamma^{\zeta_k}\E[\tau_k|\cog_k,{a}_k]{q}_k \right]$, $\Delta_1 := \{k, \zeta_k, q_k, \\ \cog_k, \tau_k\}$, $\Delta_2 := \{k, \zeta_k, q_k, \tau_k\}$, and $\Delta_3 := \{ k, \zeta_k, \tau_k\}$. 
% For brevity of notation, let $\sum_{k=0}^{n-1}\sum_{\zeta_k}\sum_{q_k}\sum_{\cog_k}\sum_{\tau_k}(\cdot)$,   $\sum_{k=0}^{n-1}\sum_{\zeta_k}\sum_{q_k}\sum_{\tau_k}(\cdot)$,   $\sum_{k=0}^{n-1}\sum_{\zeta_k}\sum_{q_k}\sum_{\cog_k}\sum_{\tau_k}(\cdot)$ be denoted by $\sum_{\Delta_1}(\cdot)$, $\sum_{\Delta_2}(\cdot)$, $\sum_{\Delta_3}(\cdot)$, respectively. 
We have
\begin{align*}
G&=\sum_{k=0}^{n-1}\E_{{\pi}}\left[ \gamma^{\zeta_k}\E[\tau_k|\cog_k,{a}_k]{q}_k \right] \\
&=\sum_{k=0}^{n-1}\E_{{\pi}}\left[ \gamma^{\zeta_k}\sum_{\tau_k}\tau_k \mathbb{P}(\tau_k|\cog_k, a_k){q}_k \right] \\
% &=\sum_{k=0}^{n-1}\sum_{\zeta_k, q_k,\cog_k}\left\{ \gamma^{\zeta_k}\E[\tau_k|\cog_k,q_k,\pi]{q}_k  \mathbb{P}( \zeta_k, \cog_k,q_k | {s}_0, {\pi}) \right\} \\
&=\sum_{\Delta_1}\left\{ \gamma^{\zeta_k} \tau_k \mathbb{P}(\tau_k|\cog_k, a_k) {q}_k  \mathbb{P}(\cog_k,q_k,  \zeta_k | {s}_0, {\pi}) \right\} \\
% &=\sum_{\Delta}\left\{ \gamma^{\zeta_k}\tau_k \mathbb{P}(\tau_k|\cog_k, q_k,\zeta_k, s_0, \pi) {q}_k \mathbb{P}( \zeta_k, \cog_k,q_k | {s}_0, {\pi}) \right\} \\
&\overset{(3)^*}{=}\sum_{\Delta_1}\left\{ \gamma^{\zeta_k} \tau_k \mathbb{P}(\tau_k, \cog_k,q_k , \zeta_k| {s}_0, {\pi}) {q}_k \right\} \\
&=\sum_{\Delta_2}\left\{ \gamma^{\zeta_k} \tau_k \mathbb{P}(\tau_k, q_k , \zeta_k| {s}_0, {\pi}) {q}_k \right\} \\
&=\sum_{\Delta_2}\left\{ \gamma^{\zeta_k}\tau_k \mathbb{P}(\tau_k , \zeta_k| {s}_0, {\pi}) \mathbb{P}( q_k| \tau_k, \zeta_k, {s}_0, {\pi}){q}_k \right\} \\
% &=\sum_{k=0}^{n-1}\left\{ \sum_{\zeta_k} \sum_{\tau_k}\gamma^{\zeta_k}\tau_k \mathbb{P}(\tau_k , \zeta_k| {s}_0, {\pi}) \sum_{q_k}{q}_k\mathbb{P}( q_k| \zeta_k, \tau_k, {s}_0, {\pi}) \right\} \\
&=\sum_{\Delta_3}\left\{ \gamma^{\zeta_k}\tau_k \mathbb{P}(\tau_k , \zeta_k| {s}_0, {\pi}) \sum_{q_k}{q}_k\mathbb{P}( q_k| \zeta_k, {s}_0, \pi) \right\} \\
&=\sum_{\Delta_3}\left\{ \gamma^{\zeta_k}\tau_k \mathbb{P}(\tau_k , \zeta_k| {s}_0, {\pi})\E_{\pi}\left[{q}_k| \zeta_k, {s}_0\right]\right\} \\
&=\sum_{k=0}^{n-1}\E_{\pi} \left[ \gamma^{\zeta_k}\tau_k \E_{\pi}\left[{q}_k| \zeta_k, {s}_0\right] \right],
\end{align*}
{where $(3)^*$ follows by using $\pi(s_k) = a_k$}.
\hfill \oprocend
\section{Proof of Lemma~\ref{lemma:value_function_bound} [\bit{Value function bounds}]}\label{app_value_function_bound}

Let $w_k$ be the number of tasks that arrive during stage $k\in \{0, \ldots, n-1\}$ with sojourn time $\tau_k$, in which the state transitions from $s_k=(q_k, \ \cog_k) \to s_{k+1}= (q_{k+1}, \ \cog_{k+1})$ and action $a_k$ is selected. Let $a_k$ be an optimal action at state $s_k$ and $\pi$ be the corresponding optimal policy such that $a_k=\pi(s_k)$. The optimal policy $\pi$ when applied from an initial state $s_0$ induces a sequence of states $<s_k>$ and  sojourn times $<\tau_k>$ or $\zeta_{k+1}$, where $\zeta_{k+1}={\sum_{j=0}^{k}\tau_j}$ and $\zeta_{0}=0$. \\
Similarly, let $\tilde{s}_0=(\tilde{q}_0,\cog_0)$ be another initial state with the same initial cognitive state, and $\tilde{q}_0 \ge q_0$. {Apply a policy $\tilde{\pi}$  from the initial state $\tilde{s}_0$ such that $\tilde{\pi}(\overline{q}, \overline{\cog})=\pi(\overline{q} + q_0-\tilde{q}_0, \overline{\cog})$ for any $(\overline{q}, \overline{\cog})$. Note that $\tilde{a}_0 = a_0.$  The optimal policy $\tilde{\pi}$ when applied from an initial state $\tilde{s}_0=(\tilde{q}_0,\cog_0)$ induces a sequence of realizations $<\tilde{s}_k>$ and $<\tilde{\tau}_k>$.
Since cognitive state and sojourn time are independent of the current queue length,   for the same action sequence applied from the initial states $s_0$ and $\tilde{s}_0$, the random process associated with the evolution of cognitive state and sojourn time is almost surely the same except for the offset in the queue length. 
Hence, the probability of observing a sequence of realizations $<\tilde{s}_k = (\tilde{q}_k , \cog_k)>$, $<\tilde{a}_k>$ and $<\tilde{\tau}_k>$ when policy $\tilde{\pi}$ is applied from $\tilde{s}_0$ is equal to the probability of observing a sequence of realizations $<{s}_k = ({q}_k, \cog_k)>$, $<a_k>$ and $<\tau_k>$ when policy ${\pi}$ is applied from ${s}_0$, where $\tilde{q}_k - \tilde{q}_0 = q_k-q_0$, $\tilde{a}_k = a_k$ and $\tilde{\tau}_k = \tau_k$. Therefore, it is easy to show that:}
\begin{equation}\label{eq: expectation_diff}
     \E_{\tilde{\pi}}[\tilde{q}_{k}| \tilde{s}_{0}, \zeta_k] -\E_{{\pi}}[q_{k}| s_{0}, \zeta_k]= \tilde{q}_0-q_0.
 \end{equation}
Note that the realization of sequence of actions $<a_k = \pi(s_k)>$, which are optimal for $<s_k=(q_k, \cog_k)>$ might be sub-optimal for $<\tilde{s}_k=(\tilde{q}_k, \cog_k)>$. Recall that $\E_{\pi}[\cdot]$ and $\E_{\tilde{\pi}}[\cdot]$ represents $\E[\cdot| s_0, \pi]$ and $\E_{\tilde{\pi}}[\cdot | \tilde{s}_0, \tilde{\pi}]$, respectively.
Let $\Delta q := \tilde{q}_0-q_0$, and $Z:={V_n^*}(q_0,\cog_0) - {V_n^*}(\tilde{q}_0,\cog_0)$. We first show the upper bound on $Z$. \\
% Similar to proof of Lemma~\ref{lemma:value_function_q}, let $a_k$ be optimal for $s_k = (q_k, \cog_k)$, and choose $<\tilde{a}_k> = <a_k>$ for the sequence $<\tilde{s}_k = (\tilde{q}_k, \cog_k)>$, where $s_0=(q_0, \cog_0)$ and $\tilde{s}_0=(\tilde{q}_0, \cog_0)$. Note that~\eqref{eq: expectation_diff} in proof of Lemma~\ref{lemma:value_function_q} still holds.
% Using $\tilde{q}_0 \ge q_0$, for $0 \le k \le n$, we have
% \begin{equation}\label{eq: expectation_diff}
%      \E[\tilde{q}_{k}| \tilde{s}_{0}, \zeta_k] -\E[q_{k}| s_{0}, \zeta_k]= \tilde{q}-q.
%  \end{equation}
$Z$ is upper-bounded by:
\begin{align}\label{eq:v_f_b}
 Z &\le {V_n^*}(q_0,\cog_0) - J_{n, \tilde{\pi}}(\tilde{q}_0,\cog_0)  \nonumber\\
&= \E_{\pi}\left[\sum_{k=0}^{n-1}\gamma^{\zeta_k}R(s_k,a_k)-\gamma^{\zeta_{n}}Cq_n \right] - \nonumber \\
& \quad \quad \quad \quad \quad \quad \E_{\tilde{\pi}}\left[\sum_{k=0}^{n-1}\gamma^{\zeta_k}R(\tilde{s}_k,\tilde{a}_k)-\gamma^{\zeta_{n}}C\tilde{q}_n \right] \nonumber \\
 &=\E_{{\pi}}\Bigg[\sum_{k=0}^{n-1}\gamma^{\zeta_k}\{r({s}_k,{a}_k)-c\E[\tau_k|{s}_k,{a}_k]{q}_k \nonumber\\ 
  &- \frac{c\lambda}{2}\E[\tau_k^2|{s}_k,{a}_k]\} -\gamma^{\zeta_{n}}C{q}_n \Bigg] -\nonumber \\
  &-\E_{\tilde{\pi}}\Bigg[\sum_{k=0}^{n-1}\gamma^{\zeta_k}\{r(\tilde{s}_k,\tilde{a}_k)-c\E[\tau_k|\tilde{s}_k,\tilde{a}_k]\tilde{q}_k \nonumber\\ 
  &- \frac{c\lambda}{2}\E[\tau_k^2|\tilde{s}_k,\tilde{a}_k]\} -\gamma^{\zeta_{n}}C\tilde{q}_n \Bigg] \nonumber \\
  &=\E_{{\pi}}\Bigg[\sum_{k=0}^{n-1}\gamma^{\zeta_k}\{r({a}_k)-c\E[\tau_k|\cog_k,{a}_k]{q}_k \nonumber\\ 
  &- \frac{c\lambda}{2}\E[\tau_k^2|\cog_k,{a}_k]\} -\gamma^{\zeta_{n}}C{q}_n \Bigg] \nonumber \\ 
  &-\E_{\tilde{\pi}}\Bigg[\sum_{k=0}^{n-1}\gamma^{\zeta_k}\{r(\tilde{a}_k)-c\E[\tau_k|\cog_k,\tilde{a}_k]\tilde{q}_k \nonumber\\ 
  &- \frac{c\lambda}{2}\E[\tau_k^2|\cog_k,\tilde{a}_k]\} -\gamma^{\zeta_{n}}C\tilde{q}_n \Bigg].
  \end{align}
 Using statements of Lemma~\ref{app:1}, RHS of \eqref{eq:v_f_b} can be written as:
\begin{align}\label{eq:v_f_q2}
&\sum_{k=0}^{n-1}\E_{{\pi}}[\gamma^{\zeta_k}r({a}_k)]  -c\sum_{k=0}^{n-1}\E_{{\pi}} \left[ \gamma^{\zeta_k}\tau_k \E_{\pi}\left[{q}_k|  {s}_0, \zeta_k \right] \right]
\nonumber \\
&- \frac{c\lambda}{2}\sum_{k=0}^{n-1}\E_{{\pi}} \left[ \gamma^{\zeta_k}\tau_k^2 \right]  -C\E_{{\pi}}\left[\gamma^{\zeta_{n}} \E_{\pi}\left[{q}_n|  {s}_0, \zeta_{n}\right] \right] \nonumber \\
&-\sum_{k=0}^{n-1}\E_{\tilde{\pi}}[\gamma^{\zeta_k}r(\tilde{a}_k)]  +c\sum_{k=0}^{n-1}\E_{\tilde{\pi}} \left[ \gamma^{\zeta_k}\tau_k \E_{\tilde{\pi}}\left[\tilde{q}_k|  \tilde{s}_0, \zeta_k \right] \right]
\nonumber \\
&+ \frac{c\lambda}{2}\sum_{k=0}^{n-1}\E_{\tilde{\pi}} \left[ \gamma^{\zeta_k}\tau_k^2 \right]  +C\E_{\tilde{\pi}}\left[\gamma^{\zeta_{n}} \E_{\tilde{\pi}}\left[\tilde{q}_n|  \tilde{s}_0, \zeta_{n}\right] \right] \nonumber \\
\overset{(4)^*}{=}&c\sum_{k=0}^{n-1}\E_{\pi}\left[\gamma^{\zeta_k}\tau_k \{\E_{\tilde{\pi}}[\tilde{q}_k \ | \tilde{s}_0, \zeta_k ]
-\E_{\pi}[q_k \ | s_0, \zeta_k ] \} \right] \nonumber \\
&+C\E_{\pi}\left[\gamma^{\zeta_{n}}\{\E_{\tilde{\pi}}[\tilde{q}_n \ | \tilde{s}_0, \zeta_k]-\E_{\pi}[q_n \ | s_0, \zeta_k ] \} \right],
\end{align}
where $(4)^*$ follows by recalling that the probability of observing a sequence of realizations $<\tilde{s}_k = (\tilde{q}_k , \cog_k)>$, $<\tilde{a}_k>$ and $<\tilde{\tau}_k>$ when policy $\tilde{\pi}$ is applied from $\tilde{s}_0$ is equal to the probability of observing a sequence of realizations $<{s}_k = ({q}_k, \cog_k)>$, $<a_k>$ and $<\tau_k>$ when policy ${\pi}$ is applied from ${s}_0$, where $\tilde{q}_k - \tilde{q}_0 = q_k-q_0$, $\tilde{a}_k = a_k$ and $\tilde{\tau}_k = \tau_k$.
Substituting~\eqref{eq: expectation_diff} in~\eqref{eq:v_f_q2}, we get,
\begin{align}\label{eq:v_f_u_b}
Z &\le \Bigg\{c\sum_{k=0}^{n-1}\E_{\pi}\left[\gamma^{\zeta_k}\tau_k\right] + C\E_{\pi}\left[\gamma^{\zeta_{n}}\right] \Bigg\}\Delta q    \nonumber\\
   &\overset{(5)^*}{=}\Bigg\{c\sum_{k=0}^{n-1}\E_{\pi}\left[\gamma^{\zeta_k}\right]\E_{\pi}\left[\tau_k \right] + C\E_{\pi}\left[\gamma^{\zeta_n}\right] \Bigg\}\Delta q   \nonumber\\
&\le   \left\{ct_{\max}\sum_{k=0}^{n-1}\rho^k  + C \rho^n \right\}\Delta q,
\end{align}
where $(5)^*$ follows due to independence of $\zeta_k=\sum_{i=0}^{k-1}\tau_k$ and $\tau_k$, $t_{\max}=\E[\tau|\cog=1,a=H]$ (assuming $\E[\tau | \cog=1, a=H] \ge \E[\tau | \cog=0, a=H]$ ), and $\rho =\max_{\cog, a}f(\E[\tau | \cog, a], \var(\tau | \cog,a))$ is an upper bound for $\E_{\pi}[\gamma^{\tau}]$ (see Assumption A3).
Taking the infinite time limit in~\eqref{eq:v_f_u_b}, we get, 
\begin{equation*}%\label{eq:v_f_u_b_f}
    V^*(q_0,\cog_0) - V^*(\tilde{q}_0,\cog_0) \le \frac{ct_{\max}\Delta q}{1 - \rho}, \ \ \  \tilde{q}_0 \ge q_0.
\end{equation*}
We now show the lower bound on $Z$. Let $\tilde{a}_k$ be optimal for $\tilde{s}_k = (\tilde{q}_k, \cog_k)$, and choose $<{a}_k> = <\tilde{a}_k>$ for the sequence $<{s}_k = ({q}_k, \cog_k)>$, where $\tilde{s}_0=(\tilde{q}_0, \cog_0)$ and $s_0=(q_0, \cog_0)$. Note that \eqref{eq: expectation_diff} still holds. Analogous to~\eqref{eq:v_f_q2}, $Z$ is lower-bounded by:
\begin{align}\label{eq:v_f_b_lower}
Z &\ge J_{n, {\pi}}({q}_0,\cog_0)- {V_n^*}(\tilde{q}_0,\cog_0)  \nonumber\\
=c &\sum_{k=0}^{n-1}\E_{\tilde{\pi}}\left[\gamma^{\tilde{\zeta}_{k}}\tilde{\tau}_k \left\{\E_{\tilde{\pi}}[\tilde{q}_k  | \tilde{s}_0, \tilde{\zeta}_{k} ]
-\E_{\pi}[q_k | s_0, \tilde{\zeta}_{k} ] \right\} \right] \nonumber \\
&+C\E_{\tilde{\pi}}\left[\gamma^{\tilde{\zeta}_{n}}\left\{\E_{\tilde{\pi}}[\tilde{q}_n | \tilde{s}_0, \tilde{\zeta}_{k}]-\E_{\pi}[q_n | s_0, \tilde{\zeta}_{k} ] \right\}\right]. \nonumber\\ 
\end{align}
Substituting~\eqref{eq: expectation_diff} in~\eqref{eq:v_f_b_lower}, we get,
\begin{align}\label{eq:v_f_u_b_lower}
Z & \ge \left\{c\sum_{k=0}^{n-1}\E_{\tilde{\pi}}\left[\gamma^{\tilde{\zeta}_{k}}\right]\E_{\tilde{\pi}}\left[\tilde{\tau}_k \right] + C\E_{\tilde{\pi}}\left[\gamma^{\tilde{\zeta}_{n}}\right] \right\}\Delta q  \nonumber\\
  &\overset{(6)^*}{\ge}  \left\{ct_s\sum_{k=0}^{n-1}\gamma^{\E_{\tilde{\pi}}\left[\tilde{\zeta}_{k}\right]} + C\gamma^{\E_{\tilde{\pi}}\left[\tilde{\zeta}_{n}\right]} \right\}\Delta q  \nonumber\\
  &\ge \left\{\frac{(1-\gamma^{nt_{\max}})ct_s}{1-\gamma^{t_{\max}}}  + \gamma^{nt_{\max}}C \right\}\Delta q,
\end{align}
where $(6)^*$ follows by applying Jensen’s inequality~\cite{dekking2005modern} ($\E[g(x)] \ge g(\E[x])$) on the convex function $g(x)= \gamma^x$. Taking the infinite time limit in~\eqref{eq:v_f_u_b_lower}, we get, 
\begin{equation*} %\label{eq:v_f_u_b_f_lower}
 0 \le \frac{ct_s\Delta q}{1-\gamma^{t_{\max}}} \le    V^*(q_0,\cog_0) - V^*(\tilde{q}_0,\cog_0) , \ \ \  \tilde{q}_0 \ge q_0.
\end{equation*}
\hfill \oprocend
\section{Proof of Lemma~\ref{lemma:threshold} [\bit{Thresholds for low cognitive states}]}\label{app_threshold1} 
 \begin{proof}
Recall that $\mc A_s := \{\{$\textit{S}, \textit{N}, \textit{H} \}$| \ s \in \mc{S}, \ q\neq 0 \}$ when queue is non-empty and $\cog \leq \cog^*$. We start by proving the first statement. In the following, we find conditions under which if action \textit{N} is the optimal choice at queue length $q_1$ for a given cognitive state $\cog \le \cog^*$, then for all $q_2>q_1$, \textit{N} dominates \textit{H}. Let \textit{N} be the optimal action in state $s_1=\{q_1,\cog\}$. Let $F(s,a)$ denote the expected future rewards  received in state $s$ for taking action $a$ (the second term in the Bellman equation~\eqref{eq:2}).
\resizebox{0.99\linewidth}{!}{
  \begin{minipage}{\linewidth}
\begin{equation}
  F(s,a) = \sum_{\cog',q',\tau}\gamma^{\tau}\mathbb{P}(q'|\tau,q)V^{*}(\cog',q') \mathbb{P}(\cog', \tau|\cog,a). \label{eq:10}  
\end{equation}
\end{minipage}
}
Then, we have
\begin{multline}{\label{eq:11new}}
    R(s_1,N)-R(s_1,H) + F(s_1,N)-F(s_1,H)>0, \\
     \implies M +
    \sum_{\tau}\sum_{\cog'}\sum_{\overline{q}}\gamma^{\tau}\texttt{Pois}(\overline{q}|\tau)V^*(\cog',q_1+\overline{q}-1) \times \\
    (\mathbb{P}(\cog',\tau | \cog, N)-\mathbb{P}(\cog',\tau | \cog, H)) > 0,
\end{multline}
where $M:=c(\mathbb{E}[\tau|\cog,H]-\mathbb{E}[\tau|\cog,N])q_1  + r_N-r_H +    \frac{c\lambda}{2}(\mathbb{E}[\tau^2|\cog,H]-\mathbb{E}[\tau^2|\cog,N])$, and $\mathbb{P}(q_1+\overline{q}-1|q_1, \tau)$ is replaced by $\texttt{Pois}(\overline{q}| \tau)$, which is the Poisson probability of $\overline{q}$ arrivals during service time $\tau$.\\
Now for the state $s_2=\{q_2,\cog\}$, with $q_2 > q_1$ and identical cognitive state $\cog$, under the assumption {(A1)} we show that:
\begin{equation}{\label{eq:12new}}
    R(s_2,N)-R(s_2,H) + F(s_2,N)-F(s_2,H)>0.
\end{equation}
The left-hand side of \eqref{eq:12new} is given by:
\begin{multline}{\label{eq:13_new}}
    X+ M +
    \sum_{\tau}\sum_{\cog'}\sum_{\overline{q}}\gamma^{\tau}\texttt{Pois}(\overline{q}|,\tau)V^*(\cog',q_2+\overline{q}-1) \times \\
    (\mathbb{P}(\cog',\tau | \cog, N)-\mathbb{P}(\cog',\tau | \cog, H)),
\end{multline}
 where $X:=c(\mathbb{E}[\tau|\cog,H]-\mathbb{E}[\tau|\cog,N])(q_2-q_1)$. To show \eqref{eq:12new}, we prove that the difference between LHS of \eqref{eq:12new} and \eqref{eq:11new} is positive. Subtracting  LHS of \eqref{eq:11new} from \eqref{eq:13_new}, we get:
\begin{multline}{\label{eq:13nnn}}
X -
\sum_{\tau}\sum_{\cog'}\sum_{\overline{q}}\gamma^{\tau}\texttt{Pois}(\overline{q}|\tau) V_D \times \\
(\mathbb{P}(\cog',\tau | \cog, N)-\mathbb{P}(\cog',\tau | \cog, H)),
\end{multline}
where $V_D:=\Big[V^*(\cog',q_1+\overline{q}-1)
- V^*(\cog',q_2+\overline{q}-1)\Big]$.
From Lemma~\ref{lemma:value_function_bound}, we know that
\begin{equation}\label{eq_manew}
   0 \le \beta := \frac{ct_s(q_2-q_1)}{1-\gamma^{t_{\max}}} \le V_D \le \frac{ct_{\max}(q_2-q_1)}{1 - \rho} =: \alpha.
\end{equation}
Therefore, \eqref{eq:13nnn} is lower bounded by
\begin{multline}{\label{eq:144new}}
X +
\beta\sum_{\tau}\sum_{\cog'}\sum_{\overline{q}}\gamma^{\tau}\texttt{Pois}(\overline{q}|\tau)\mathbb{P}(\cog',\tau | \cog, H) \\ -\alpha\sum_{\tau}\sum_{\cog'}\sum_{\overline{q}}\gamma^{\tau}\texttt{Pois}(\overline{q}|\tau)\mathbb{P}(\cog',\tau | \cog, N)  \\
\ge X + \beta\gamma^{\E[\tau|\cog,H]} - \alpha\E[\gamma^{\tau}| \cog , N],
\end{multline}
where we utilized Jensen's inequality on convex function $\gamma^{\tau}$ to obtain $ \E[\gamma^\tau| \cog, H] \ge \gamma^{\E[\tau|\cog,H]}$. \eqref{eq:144new} is non negative for
\begin{multline}{\label{eq:164new}}
\mathbb{E}[\tau|\cog,H]-\mathbb{E}[\tau|\cog,N] + \frac{t_s\gamma^{\E[\tau|\cog,H]}}{1-\gamma^{t_{\max}}} \\ \ge \frac{t_{\max}}{1 - \rho}\E[\gamma^{\tau}| \cog , N],
\end{multline}
where $t_{\max} =\E[\tau | \cog=1, a=H] $ is the maximum expected sojourn time, and $t_s$ is the constant time for skip.\\
% \red{Note that for tasks which require large service time, i.e., $0 \ll \E[\tau| \cog, N] \ll \E[\tau| \cog, H]$, $\E[\gamma^{\tau}| \cog , N] \rightarrow 0$, and~\eqref{eq:164new} always holds.} \\
Now we prove the second statement. Using a similar analysis it can be shown that if action \textit{S} is the optimal choice at queue length $q_1$ for a given cognitive state $\cog \le \cog^*$, then for every $q_2>q_1$, \textit{S} dominates \textit{H} and \textit{N} under the following conditions:
\begin{multline}{\label{eq:165new}}
\mathbb{E}[\tau|\cog,H]-t_s+ \frac{t_{s}\gamma^{\E[\tau|\cog,H]}}{1-\gamma^{t_{\max}}} \ge \gamma^{t_s}\frac{t_{\max}}{1 - \rho},
\end{multline}
\begin{multline}{\label{eq:166new}}
\mathbb{E}[\tau|\cog,N]-t_s + \frac{t_{s}\gamma^{\E[\tau|\cog,N]}}{1-\gamma^{t_{\max}}} \ge \gamma^{t_s}\frac{t_{\max}}{1 - \rho},
\end{multline}
respectively, where we have used $\mathbb{E}[\tau|\cog,S] =t_s$ and $\E[\gamma^{\tau}| \cog , S]=\gamma^{t_s}$ due to constant time of skip. Since $\mathbb{E}[\tau|\cog,H] > \mathbb{E}[\tau|\cog,N]$, 
% and {$\E[\gamma^{\tau}| \cog , N] > \E[\gamma^{\tau}| \cog , H]$},
\eqref{eq:165new}-\eqref{eq:166new} can be combined to obtain the condition:
\begin{multline}{\label{eq:167new}}
\mathbb{E}[\tau|\cog,N]-t_s + \frac{t_s\gamma^{\E[\tau|\cog,H]}}{1-\gamma^{t_{\max}}} \ge \gamma^{t_s}\frac{t_{\max}}{1 - \rho},
\end{multline}
under which action \textit{S} dominates both \textit{H} and \textit{N}.
\end{proof}

\begin{wrapfigure}{l}{1in}
\includegraphics[width=1in,height=1.25in,clip,keepaspectratio]{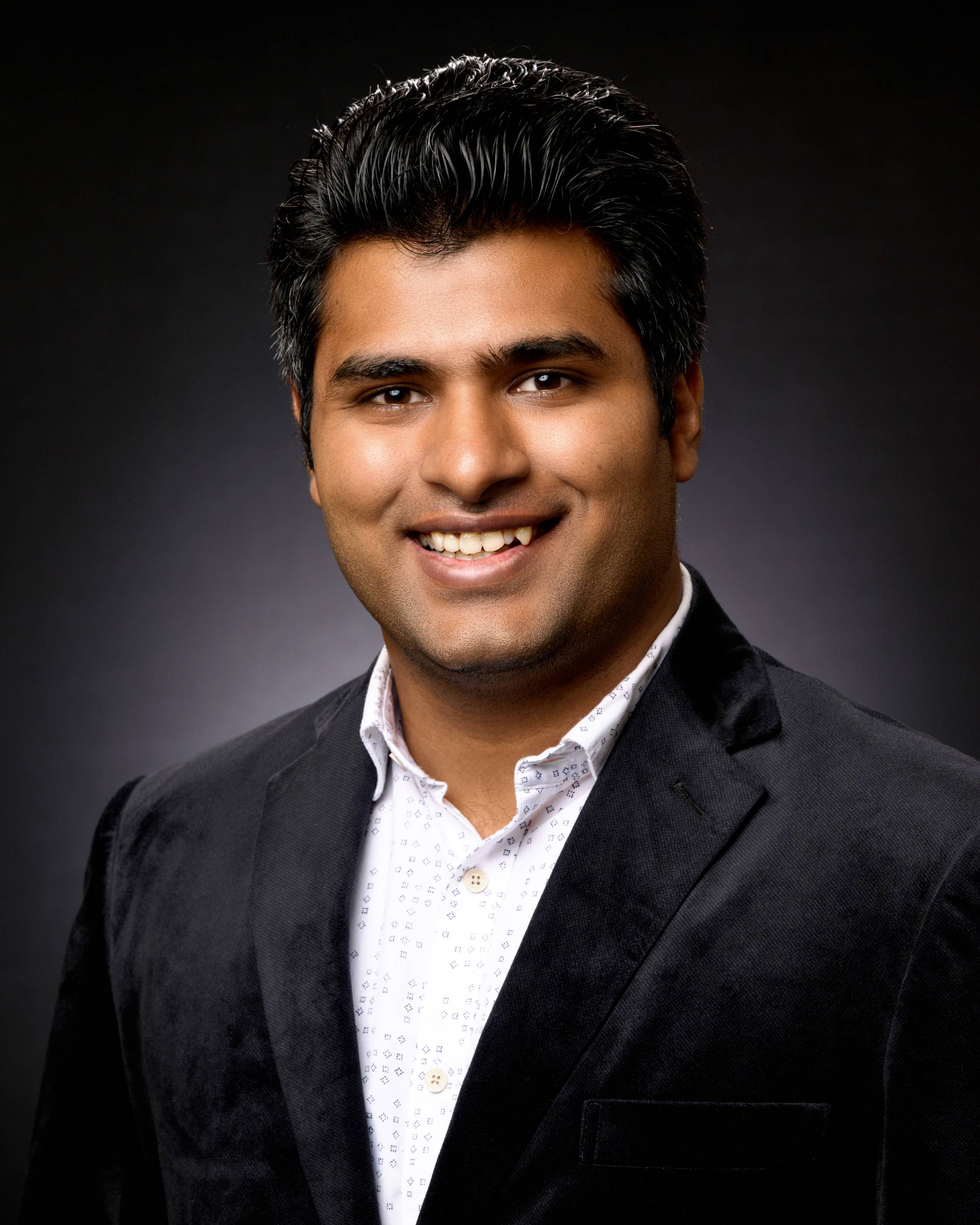}
\end{wrapfigure}

\noindent
\textbf{Piyush Gupta} is currently a Ph.D. candidate in the Department of Electrical and Computer Engineering at Michigan State University. He earned his B.Tech. degree in Mechanical Engineering from the Indian Institute of Technology, Delhi, India, in 2015. During the years 2015 to 2017, he served as an R\&D Engineer at Honda R\&D Co. Ltd., Japan. Subsequently, in 2020, he completed his M.S. degree in Electrical and Computer Engineering at Michigan State University.
His research interests encompass a variety of areas, including human-in-the-loop systems, motion planning and prediction for autonomous vehicles, and machine learning algorithms.

\begin{wrapfigure}{l}{1in}
  \includegraphics[width=1in,height=1.25in,clip,keepaspectratio]{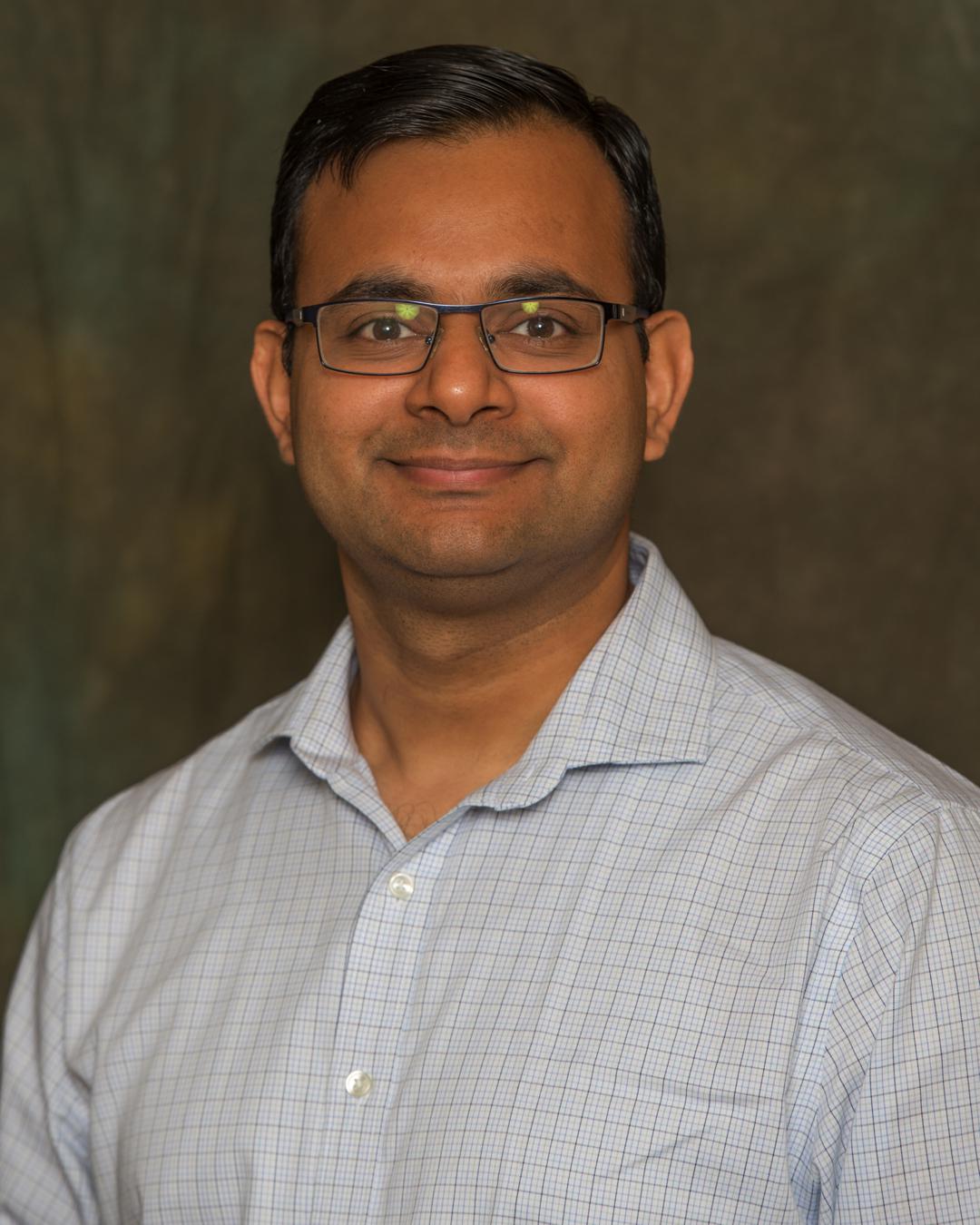}
\end{wrapfigure}

\noindent
\textbf{Vaibhav Srivastava} received the B.Tech. degree (2007) in mechanical engineering from the Indian Institute of Technology Bombay, Mumbai, India; the M.S. degree in mechanical engineering (2011), the M.A. degree in statistics (2012), and the Ph.D. degree in mechanical engineering (2012) from the University of California at Santa Barbara, Santa Barbara, CA.\\
Dr. Srivastava is currently an Associate Professor of Electrical and Computer Engineering at Michigan State University. He is also affiliated with Mechanical Engineering, Cognitive Science Program, and Connected and Autonomous Networked Vehicles for Active Safety (CANVAS). He served as a Lecturer and Associate Research Scholar with the Mechanical and Aerospace Engineering Department at Princeton University, Princeton, NJ from 2013-2016. %He has served on the IEEE Control System Society conference editorial board since 2018. He received the best paper award (as coauthor) at the 2014 European Control Conference. 
His research focuses on Cyber Physical Human Systems with an emphasis on mixed human-robot systems and networked multi-agent systems. 

\end{document}